\documentclass[9pt,twocolumn]{article}

\usepackage{dsfont}

\usepackage{url}

\usepackage{amsmath} 

\usepackage{amsthm} 

\usepackage{amssymb} 

\usepackage{graphicx,epsfig}

\usepackage{pst-tree,pstricks}

\newcommand{\eat}[1]{}

\newcommand{\nc}{\newcommand}

\def\extraspacing{\vspace{2mm}}

\nc{\cH}{{\mathcal H}}
\nc{\cO}{{\mathcal O}}
\nc{\cT}{{\mathcal T}}

\nc{\logspace}{\text{LOGSPACE}\xspace}
\nc{\logdcfl}{\text{LOGDCFL}\xspace}
\nc{\logcfl}{\text{LOGCFL}\xspace}
\nc{\nlogspace}{\text{NLOGSPACE}\xspace}
\nc{\ptime}{\text{PTIME}\xspace}
\nc{\np}{\text{np}\xspace}
\nc{\conp} {\textrm{co}\text{NP}\xspace}
\nc{\pspace}{\text{PSPACE}\xspace}
\nc{\exptime}{\text{EXPTIME}\xspace}
\nc{\nexptime}{\text{NEXPTIME}\xspace}

\newtheorem{example}{Example}
\newtheorem{theorem}{Theorem}

\newtheorem{definition}{Definition}
\newtheorem{corollary}{Corollary}
\newtheorem{lemma}{Lemma}
\newtheorem{proposition}{Proposition}

\nc{\SL}{SL}
\nc{\San}{\mathcal{A}}
\nc{\SanC}{\text{San-Clicks}}
\nc{\Sen}{\text{Sen}}

\newcommand{\nop}[1]{}

\begin{document}
\begin{sloppy}

\title{Publishing Search Logs -- A Comparative Study of Privacy Guarantees}

\author{
Michaela G\"otz\\
\and
Ashwin Machanavajjhala\\
\and
Guozhang Wang\\
\and
Xiaokui Xiao\\
\and
Johannes Gehrke\\
}

\maketitle

\begin{abstract}

Search engine companies collect the ``database of intentions'', the
histories of their users' search queries. These search logs are a gold
mine for researchers. Search engine companies, however, are wary of
publishing search logs in order not to disclose sensitive information.

In this paper we analyze algorithms for publishing frequent keywords,
queries and clicks of a search log.  We first show how methods that
achieve variants of $k$-anonymity are vulnerable to active attacks. We
then demonstrate that the stronger guarantee ensured by
$\epsilon$-differential privacy unfortunately does not provide any
utility for this problem. We then propose an algorithm ZEALOUS
and show how to set its parameters to achieve
$(\epsilon,\delta)$-probabilistic privacy. We also contrast our
analysis of ZEALOUS with an analysis by Korolova et
al.~\cite{KorolovaKMN09:PrivateQueries} that achieves
$(\epsilon',\delta')$-indistinguishability.

Our paper concludes with a large experimental study using real
applications where we compare ZEALOUS and previous work that achieves
$k$-anonymity in search log publishing. Our results show that ZEALOUS
yields comparable utility to $k-$anonymity while at the same time
achieving much stronger privacy guarantees.
\end{abstract}

\section{Introduction}

\vspace{2ex}
\emph{
Civilization is the progress toward a society of privacy. The savage's whole existence is public, ruled by the laws of his tribe. Civilization is the process of setting man free from men.} --- Ayn Rand.\\
\indent
\emph{My favorite thing about the Internet is that you get to go into the private world of real creeps without having to smell them.}
--- Penn Jillette.
\vspace{2ex}

Search engines play a crucial role in the navigation through the
vastness of the Web.  Today's search engines do not just collect and
index webpages, they also collect and mine information about their
users.  They store the queries, clicks, IP-addresses, and other
information about the interactions with users in what is called a {\em
search log}.
Search logs contain valuable information that search engines use to
tailor their services better to their users' needs. They enable the
discovery of trends, patterns, and anomalies in the search behavior of
users, and they can be used in the development and testing of new
algorithms to improve search performance and quality.  Scientists all
around the world would like to tap this gold mine for their own
research; search engine companies, however, do not release them
because they contain sensitive information about their users, for
example searches for diseases, lifestyle choices, personal tastes, and
political
affiliations.


The only release of a search log happened in 2006 by AOL, and it went
into the annals of tech history as one of the great debacles in the
search industry.\footnote{
\url{http://en.wikipedia.org/wiki/AOL_search_data_scandal}
describes the incident, which resulted in the resignation of AOL's CTO
and an ongoing class action lawsuit against AOL resulting from the
data release.}  AOL published three months of search logs of
$650\mathord{,}000$ users.  The only measure to protect user privacy
was the replacement of user--ids with random numbers --- utterly
insufficient protection as the New York Times showed by identifying a
user from Lilburn, Georgia~\cite{NYTimes:AOL}, whose search queries
not only contained identifying information but also sensitive
information about her friends' ailments.
%

The AOL search log release shows that simply replacing
user--ids with random numbers does not prevent information disclosure. Other ad--hoc methods have been studied and found to be similarly insufficient,
such as the removal of names, age, zip codes and other
identifiers~\cite{JonesKPT07:PrivacyQueryLogs} and the replacement of
keywords in search queries by random
numbers~\cite{KumarNPT07:AnonQueryLogs}.

In this paper, we compare  formal methods of limiting disclosure when publishing frequent keywords, queries, and clicks of a search log. The methods 
vary in the guarantee of disclosure limitations they provide and in
the amount of useful information they retain.
We first describe two negative results. We show that existing proposals
to achieve $k$-\emph{anonymity}~\cite{Samarati01} in search
logs~\cite{Adar07:AnonQueryLogs,MotwaniN:searchlogs,HeN09:HierarchicalAnon,YuanAnonymization09}
are insufficient in the light of attackers who can actively
  influence the search log.
We then turn to
\emph{differential privacy}~\cite{dwork06Calibrating}, a much stronger
privacy guarantee; however, we show that it is impossible to achieve
good utility
with differential
privacy.

We then describe Algorithm ZEALOUS\footnote{{\bf ZEA}rch
  {\bf LO}g p{\bf U}bli{\bf S}ing}, developed independently by Korolova
et al.~\cite{KorolovaKMN09:PrivateQueries} and
us~\cite{GoetzMWXG09:searchlogsV2} with the goal to achieve relaxations
of differential privacy.  Korolova et al.~showed how to set the
parameters of ZEALOUS to guarantee $(\epsilon,
\delta)$-indistinguishability~\cite{DworkKMMN06_OurDataOurselves}, and
we here offer a new analysis that shows how to set the parameters of
ZEALOUS to guarantee $(\epsilon, \delta)$\emph{-probabilistic
  differential privacy}~\cite{ashwin08:map}
(Section~\ref{sec:alg:prob}), a much stronger privacy guarantee as our
analytical comparison shows.

Our paper concludes with an extensive experimental evaluation, where
we compare the utility of various algorithms that guarantee anonymity
or privacy in search log publishing.
Our
evaluation includes applications that use search logs for improving
both search experience and search performance, and our results show
that ZEALOUS' output is sufficient for these applications
while achieving strong formal privacy guarantees.

We believe that the results of this research enable search engine
companies to make their search log available to researchers without
disclosing their users' sensitive information: Search engine companies
can apply our algorithm
to generate statistics that are
$(\epsilon, \delta)$-probabilistic differentially private while
retaining good utility for the two applications we have tested.
Beyond publishing search logs we believe
that our findings are of interest when publishing frequent
itemsets, as ZEALOUS protects privacy against much stronger
attackers than those considered in existing work on privacy-preserving
publishing of frequent items/itemsets \cite{Luo:SurveyItemsets09}.

The remainder of this paper is organized as follows. We start with
some background in Section \ref{sec:prelim}. Our
negative results are presented in Section~\ref{sec:negative}. We then describe Algorithm ZEALOUS and its analysis in Section \ref{sec:prob_privacy}.  We
compare indistinguishability with probabilistic differential privacy
in Section \ref{sec:comp}. Section \ref{sec:exp_para} shows the
results of an extensive study of how to set parameters in ZEALOUS, and
Section \ref{sec:apps} contains a thorough evaluation of ZEALOUS in
comparison with previous work. We conclude with a discussion of
related work and other applications.

\section{Preliminaries} \label{sec:prelim}

In this section we introduce the problem of publishing frequent
keywords, queries, clicks and other items of a search log.

\subsection{Search Logs}

Search engines such as Bing, Google, or Yahoo log interactions with
their users.  When a user submits a query and clicks on one or more
results, a new entry is added to the search log. Without loss of
generality, we assume that a search log has the following schema:
\[
\langle \textsc{user-id,  query, time, clicks}\rangle,
\]
where a \textsc{user-id} identifies a user, a \textsc{query} is a set
of keywords, and \textsc{clicks} is a list of \emph{urls} that the
user clicked on.  The user-id can be determined in various ways; for
example, through cookies, IP addresses or user accounts.  A \emph{user
  history} or \emph{search history} consists of all search entries
from a single user. Such a history is usually partitioned into
\emph{sessions} containing similar queries; how this partitioning is
done is orthogonal to the techniques in this paper.  A \emph{query
  pair} consists of two subsequent queries from the same user within the same session.
%

We say that a user history \emph{contains} a keyword $k$ if there
exists a search log entry such that $k$ is a keyword in the query of
the search log.  A \emph{keyword histogram} of a search log $S$
records for each keyword $k$ the number of users $c_k$ whose search
history in $S$ contains $k$. A keyword histogram is thus a set of
pairs $(k,c_k)$. We define the \emph{query histogram}, the \emph{query
  pair histogram}, and the \emph{click histogram} analogously.  We
classify a keyword, query, consecutive query, click in a histogram to
be \emph{frequent} if its count exceeds some predefined threshold
$\tau$; when we do not want to specify whether we count keywords,
queries, etc., we also refer to these objects as \emph{items}.

With this terminology, we can define our goal as publishing frequent
items (utility) without disclosing sensitive information about the
users (privacy). We will make both the notion of utility and privacy
more formal in the next sections.


\subsection{Disclosure Limitations for Publishing Search Logs}

A simple type of disclosure is the identification of a particular
user's search history (or parts of the history) in the published search
log. The concept of $k$-anonymity has been introduced to avoid such
identifications.
\extraspacing
\begin{definition}[$k$-anonymity~\cite{Samarati01}] \label{def:anonymity}
  \em A search log is $k$-anonymous if the search history of every
  individual is indistinguishable from the history of at least $k-1$
  other individuals in the published search log.
\end{definition}
\extraspacing
There are several proposals in the literature to achieve different
variants of $k$-anonymity for search logs. Adar proposes to partition
the search log into sessions and then to discard queries that are
associated with fewer than $k$ different user-ids. In each session the
user-id is then replaced by a random
number~\cite{Adar07:AnonQueryLogs}. We call the output of Adar's
Algorithm a \emph{$k$-query anonymous search log}.
Motwani and Nabar add or delete keywords from sessions until each
session contains the same keywords as at least $k-1$ other sessions in
the search log~\cite{MotwaniN:searchlogs}, following by a replacement of the
user-id by a random number. We call the output of this algorithm a
\emph{$k$-session anonymous search log}.  He and Naughton generalize
keywords by taking their prefix until each keyword is part of at least
$k$ search histories and publish a histogram of the partially
generalized keywords \cite{HeN09:HierarchicalAnon}. We call the output
a \emph{$k$-keyword anonymous search log}. Efficient ways to anonymize a
search log are
also discussed by Yuan
et al.~\cite{YuanAnonymization09}.


Stronger disclosure limitations try to limit what an attacker can learn about a user.
Differential privacy guarantees that an attacker learns roughly the
same information about a user whether or not the search history of
that user was included in the search log~\cite{dwork06Calibrating}.
Differential privacy
has previously been applied to contingency
tables~\cite{barakCDKMT07:holistic}, learning
problems~\cite{BlumLR08,KasiviswanathanLNRS08}, synthetic data
generation~\cite{ashwin08:map} and more.

\extraspacing
\begin{definition}[$\epsilon$-differential privacy \cite{dwork06Calibrating}] \label{def:ediff}  \em
An algorithm $\mathcal{A}$ is $\epsilon$-differentially private if for
all search logs $S$ and $S'$ differing in the search history of a
single user and for all  output search logs $O$:
\[Pr[\mathcal{A}(S)=O] \leq e^\epsilon  Pr[\mathcal{A}(S')=O].\]
\end{definition}
\extraspacing
This definition ensures that the output of the algorithm is insensitive to changing/omitting the complete search history of a single user.
We will refer to search logs that only differ in the search history of
a single user as \emph{neighboring search logs}.
Similar to the variants of $k$-anonymity we could also
define variants of differential privacy by looking at neighboring
search logs that differ only in the content of one session, one query
or one keyword. However, we chose to focus on the strongest definition
in which an attacker learns roughly the same about a user even if that user's whole search history was omitted.

Differential privacy is a very strong guarantee and in some cases it
can be too strong to be practically achievable.
%
We will review two relaxations that have been proposed in the literature.
 Machanavajjhala et
al.~proposed the following probabilistic version of differential
privacy.
\extraspacing
\begin{definition} [probabilistic differential privacy \cite{ashwin08:map}]
\label{def:prob-diff} \em
An algorithm $\mathcal{A}$ satisfies $(\epsilon,\delta)$-probabilistic
differential privacy if for all search logs $S$ we can divide the
output space $\Omega$ into two sets $\Omega_1, \Omega_2$ such that
\[(1) \Pr[\mathcal{A}(S) \in \Omega_2] \leq  \delta, \text{ and }\]
 for all neighboring search logs $S'$ and for all $O \in \Omega_1$:
\begin{align*}
(2) &  e^{-\epsilon} \Pr[\mathcal{A}(S')=O] \leq \Pr[\mathcal{A}(S)=O] \leq  e^{\epsilon} \Pr[\mathcal{A}(S')=O] 
\end{align*}
\end{definition}
\extraspacing
This definition guarantees that algorithm $\mathcal{A}$ achieves
$\epsilon$-differential privacy with high probability ($\geq 1-
\delta$). The set $\Omega_2$ contains all outputs that are considered
privacy breaches according to $\epsilon$-differential privacy; the
probability of such an output is bounded by $\delta$.

The following relaxation
has been proposed by
Dwork et al.\ \cite{DworkKMMN06_OurDataOurselves}. 
\extraspacing
\begin{definition} [indistinguishability \cite{DworkKMMN06_OurDataOurselves}] \label{def:indist}
  \em An algorithm $\mathcal{A}$ is $(\epsilon,\delta)$-indistinguishable
  if for all search logs $S, S'$ differing in one user history and for
  all subsets $\mathcal{O}$ of the output space $\Omega$:
	\[ \Pr[\mathcal{A}(S) \in \mathcal{O}] \leq e^{\epsilon}\Pr[\mathcal{A}(S') \in \mathcal{O}] + \delta \]
\end{definition}
\extraspacing
We will compare these two definitions in Section~\ref{sec:comp}. In
particular, we will show that probabilistic differential privacy
implies indistinguishability, but the converse does not
hold:  We show that there exists an algorithm that is $(\epsilon',
\delta')$-indistinguishable yet 
not $(\epsilon,\delta)$-probabilistic differentially private for any
$\epsilon$ and any $\delta<1$, thus showing that
$(\epsilon,\delta)$-probabilistic differential privacy is clearly
stronger than $(\epsilon', \delta')$-indistinguishability.

\subsection{Utility Measures}

We will compare the utility of algorithms producing sanitized search
logs both theoretically and experimentally.

\subsubsection{Theoretical Utility Measures}\label{sec:utility_def}
For simplicity, suppose we want to publish all items (such as
keywords, queries, etc.) with frequency at least $\tau$ in a search
log; we call such items \emph{frequent items}; we call all other items
\emph{infrequent items}.
Consider a discrete domain of items $\mathcal{D}$. Each user
contributes a set of these items to a search log $S$.  We denote by
$f_d(S)$ the frequency of item $d \in \mathcal{D}$ in search log
$S$. We drop the dependency from $S$ when it is clear from the
context.

We define the inaccuracy of a (randomized) algorithm as the expected
number of items it gets wrong, i.e., the number of frequent items that
are not included in the output, plus the number of infrequent items
that are included in the output.
We do not expect an algorithm to be perfect. It may make mistakes for
items with frequency very close to $\tau$, and thus we do not take
these items in our notion of accuracy into account.  We formalize this
``slack'' by a parameter $\xi$, and given $\xi$, we introduce the
following new notions. We call an item $d$ with frequency $f_d \geq
\tau + \xi$ a \emph{very-frequent item} and an item $d$ with frequency
$f_d \leq \tau - \xi$ a \emph{very-infrequent item}. We will measure
the inaccuracy of an algorithm then only using its inability to retain
the very-frequent items and its inability to filter out the very
infrequent items.
%
%
%
%
\extraspacing
\begin{definition} [$(\mathcal{A}, S)$-inaccuracy] \label{def:acc} \em
  Given an algorithm $\mathcal{A}$ and an input search log $S$, the $(\mathcal{A},
  S)$-inaccuracy with slack $\xi$ is defined as
\begin{align*}
E[ | & \{ d \in \mathcal{A}(S) | f_d(S) < \tau - \xi \} \cup \\
	 & \{ d \not \in \mathcal{A}(S) | f_d(S) > \tau + \xi \}  |]
\end{align*}
\end{definition}
\extraspacing
The expectation is taken over the randomness of the algorithm.  As an
example, consider the simple algorithm that always outputs the empty
set; we call this algorithm the \emph{baseline algorithm}. On input
$S$ the Baseline Algorithm has an inaccuracy equal to the number of
items with frequency at least $\tau + \xi$.

For the results in the next sections it will be useful to distinguish
the error of an algorithm on the very-frequent items and its error on
the very-infrequent items.  We can rewrite the inaccuracy as:
\begin{align*}
	  \hspace{-0.3cm}   \sum_{d: f_d(S)>\tau+\xi}
          \hspace{-0.5cm}1-\Pr[d \in \mathcal{A}(S)] 	
          + \hspace{-0.8cm}\sum_{d  \in D: f_d(S)<\tau-\xi}
          \hspace{-0.8cm}\Pr[d \in \mathcal{A}(S)]
\end{align*}
Thus, the $(\mathcal{A}, S)$-inaccuracy with slack $\xi$ can be
rewritten as the inability to retain the very-frequent items plus the
inability to filter out the very-infrequent items.  For example, the
baseline algorithm has an inaccuracy to filter of 0 inaccuracy to retain equal to the number of very-frequent items.

\extraspacing
\begin{definition} [$c$--accuracy] \label{def:accuracy_frequent} \em
  An algorithm $\mathcal{A}$ is $c$--accurate 
  if for any input search log $S$ and any very-frequent item $d$ in
  $S$, the probability that $\mathcal{A}$ outputs $d$ is at least $c$.
\end{definition}

\subsubsection{Experimental Utility Measures} \label{sec:expUtilM}

Traditionally, the utility of a privacy-preserving algorithm has been
evaluated by comparing some statistics of the input with the output to
see ``how much information is lost.'' The choice of suitable
statistics is a difficult problem as these statistics need to mirror
the sufficient statistics of applications that will use the sanitized
search log, and for some applications the sufficient statistics are
hard to characterize. To avoid this drawback, Brickell et
al.~\cite{BrickellS08:DataMiningUtility} measure the utility with
respect to data mining tasks and they take the actual classification
error of an induced classifier as their utility metric.

In this paper we take a similar approach. We use two real applications
from the information retrieval community: Index caching, as a
representative application for search performance, and query
substitution, as a representative application for search quality.  For
both application the sufficient statistics are histograms of
keywords, queries, or query pairs.

\extraspacing
\noindent\textbf{Index Caching.}
Search engines maintain an inverted index which, in its simplest
instantiation, contains for each keyword a posting list of identifiers
of the documents in which the keyword appears.  This index can be used
to answer search queries, but also to classify queries for choosing
sponsored search results.  The index is usually too large to fit in
memory, but maintaining a part of it in memory reduces response time
for all these applications. We use the formulation of the \emph{index
  caching problem} from Baeza--Yates~\cite{Baeza04}. We are given a
keyword search workload, a distribution over keywords indicating the
likelihood of a keyword appearing in a search query. It is our goal to
cache in memory a set of
posting lists that for a given workload maximizes the
cache-hit-probability while not exceeding the
storage capacity.  Here the hit-probability is the probability that
the posting list of a keyword can be found in memory given the keyword search workload.
%

\extraspacing
\noindent\textbf{Query Substitution.}
Query substitutions are suggestions to rephrase a user query to match
it to documents or advertisements that do not contain the actual
keywords of the query.
Query substitutions can be applied in query refinement, sponsored
search, and spelling error
correction~\cite{JonesRMG06:QuerySubstitution}.  Algorithms for query
substitution examine query pairs to learn how users re-phrase
queries. We use an algorithm developed by Jones et
al.~\cite{JonesRMG06:QuerySubstitution}.



\section{Negative Results}
\label{sec:negative}

In this section, we discuss the deficiency of two existing models of disclosure limitations for search log publication. Section~\ref{sec:anon} focuses on $k$-anonymity, and Section~\ref{sec:diff_privacy} investigates differential privacy.

\subsection{Insufficiency of Anonymity}
\label{sec:anon}

$k$-anonymity and its variants prevent an attacker from uniquely identifying the user that corresponds to a search history in the sanitized search log. While it offers great utility even beyond releasing frequent items its disclosure guarantee might not be satisfactory.  Even without unique identification of a user, an attacker can infer the keywords or queries used by the user. $k$-anonymity does not protect against this severe information disclosure.

There is another issue largely overlooked with the current implementations of anonymity. That is instead of guaranteeing that the keywords/queries/sessions of $k$ \emph{individuals} are indistinguishable in a search log they only assure that the keywords/queries/sessions associated with $k$ different \emph{user-IDs} are indistinguishable. 
These two guarantees are not the same since individuals can have multiple accounts or share accounts. 
An attacker can exploit this by creating multiple accounts and submitting the same fake queries from these accounts. It can happen that in a $k$-keyword/query/session-anonymous search log the  keywords/queries/sessions of a user  are only indistinguishable from $k-1$ fake keywords/queries/sessions submitted by an attacker. It is doubtful that this type of indistinguishability at the level of user-IDs is satisfactory.

\subsection{Impossibility of Differential Privacy}
\label{sec:privacy}\label{sec:diff_privacy}

In the following, we illustrate the infeasibility of differential
privacy in search log publication. In particular, we show
that, under realistic settings, no differentially private algorithm can
produce a sanitized search log with reasonable utility (utility is measured as defined in Section~\ref{sec:utility_def} using our notion of accuracy). 
Our analysis is based on
the following lemma.

\extraspacing
\begin{lemma}\label{thm:inacc_per_item} \em
  For a set of $U$ users, let $S$ and $S'$ be two search logs each containing
  at most $m$ items from some domain $D$  per user.
  Let $\mathcal{A}$ be
  an $\epsilon$-differentially private algorithm that, given $S$,
  retains a very-frequent item $d$ in $S$ with probability $p$. Then,
  given $S'$, $\mathcal{A}$ retains $d$ with probability at least $p /
  (e^{L_1(S, S')\cdot\epsilon/m})$, where $L_1(S, S') = \sum_{d \in D}
  |f_d(S) - f_d(S')|$ denotes the $L_1$ distance between $S$ and $S'$.
\end{lemma}
\extraspacing

Lemma~\ref{thm:inacc_per_item} follows directly from the definition of
$\epsilon$-differential privacy. Based on
Lemma~\ref{thm:inacc_per_item}, we have the following theorem, which
shows that any $\epsilon$-differentially private algorithm that is
accurate for very-frequent items must be inaccurate for
very-infrequent items. The rationale is that, if given a search log
$S$, an algorithm outputs one very-frequent item $d$ in $S$, then even
if the input to the algorithm is a search log where $d$ is
very-infrequent, the algorithm should still output $d$ with a certain
probability; otherwise, the algorithm cannot be differentially
private.

\extraspacing
\begin{theorem}\label{thm:inacc} \em
  Consider an accuracy constant $c$, a threshold $\tau$, a slack $\xi$
  and a very large domain $\mathcal{D} $ of size $\geq U m
  \left(\frac{2 e^{2\epsilon(\tau+\xi)/m}}{c(\tau+\xi)}
    +\frac{1}{\tau-\xi+1}\right)$, where $m$ denotes the maximum
  number of items that a user may have in a search log. Let
  $\mathcal{A}$ be an $\epsilon$-differentially private algorithm that
  is $c$-accurate (according to Definition
  \ref{def:accuracy_frequent}) for the very-frequent items. Then, for
  any input search log, the inaccuracy of $\mathcal{A}$ is greater
  than the inaccuracy of an algorithm that always outputs an empty
  set.
\end{theorem}
\begin{proof}
Consider an $\epsilon$-differentially private algorithm $\mathcal{A'}$ that is $c$-accurate for the very-frequent items. Fix some input $S$. We are going to show that for each very-infrequent item $d$ in $S$ the probability of outputting $d$ is at least $c/(e^{\epsilon(\tau+\xi)/m})$.
For each item $d \in \mathcal{D}$ construct $S_d'$ from $S$ by changing $\tau+\xi$ of the items to $d$. That way  $d$ is very-frequent (with frequency at least $\tau+\xi$) and $L_1(S,S_d')\leq 2(\tau+\xi)$. By Definition~\ref{def:accuracy_frequent}, we have that
\[
\Pr[d \in \mathcal{A'}(S_d')] 	\geq c.
\]
By Lemma~\ref{thm:inacc_per_item} it follows that the probability of outputting $d$ is at least  $c/(e^{2\epsilon(\tau+\xi)/m})$ for any input database.
This means that we can compute a lower bound on the inability to filter out the very-infrequent items in $S$ by  summing up this probability over all possible values $d \in \mathcal{D}$ that are very-infrequent in $S$. Note, that there are at least $\mathcal{D} - \frac{U m}{\tau-\xi+1}$ many very-infrequent items in $S$.
Therefore, the inability to filter out the very-infrequent items is at least $\left(|\mathcal{D}|-\frac{U m}{\tau-\xi+1}\right)c/(e^{2\epsilon(\tau+\xi)/m})$.
For large domains of size at least $ U m \left(\frac{2 e^{2\epsilon(\tau+\xi)/m}}{c(\tau+\xi)} +\frac{1}{\tau-\xi+1}\right)$ the inaccuracy is at least $\frac{2 U m}{\tau+\xi}$ which is greater than the inaccuracy of the baseline.
\end{proof}
\extraspacing


To illustrate Theorem~\ref{thm:inacc}, let us consider a search log $S$ where each query contains at most 3 keywords selected from a limited vocabulary of 900,000 words. Let $D$ be the domain of the consecutive query pairs in $S$. We have $|D| = 5.3 \times 10^{35}$. Consider the following setting of the parameters $\tau+\xi = 50, m = 10, U = 1\mathord{,}000\mathord{,}000, \epsilon =1$, that is typical practice. By Theorem~\ref{thm:inacc}, if an $\epsilon$-differentially private algorithm $\mathcal{A}$ is $0.01$-accurate for very-frequent query pairs, then, in terms of overall inaccuracy (for both very-frequent and very-infrequent query pairs), $\mathcal{A}$ must be inferior to an algorithm that always outputs an empty set. In other words, no differentially private algorithm can be accurate for both very-frequent and very-infrequent query pairs.

\section{Achieving Privacy}\label{sec:alg}
\label{sec:prob_privacy}

In this section, we introduce a search log publishing algorithm called ZEALOUS that has been
independently developed by Korolova et al.\ \cite{KorolovaKMN09:PrivateQueries} and
us~\cite{GoetzMWXG09:searchlogsV2}. ZEALOUS ensures probabilistic differential privacy, and it follows a simple two-phase framework. In the first phase, ZEALOUS generates a histogram of items in the input search log, and then removes from the histogram the items with frequencies below a threshold. In the second phase, ZEALOUS adds noise to the histogram counts, and eliminates the items whose noisy frequencies are smaller than another threshold. The resulting histogram (referred to as the {\em sanitized} histogram) is then returned as the output. Figure~\ref{fig:alg} depicts the steps of ZEALOUS.


\extraspacing
\noindent\textbf{Algorithm ZEALOUS} for Publishing Frequent Items of a Search Log \\
\noindent\textbf{Input:}  Search log $S$, positive numbers $m$, $\lambda$, $\tau$, $\tau'$
\begin{enumerate}
\item[1.] 
 For each user $u$ select a set $s_u$ of  up to $m$ distinct items from  $u$'s search history in $S$.\footnote{These items can be selected in various ways as long as the selection criteria is not based on the data. Random selection is one candidate.}
\item[2.] Based on the selected items, create a histogram consisting of pairs $(k, c_k)$, where $k$ denotes an item and $c_k$ denotes the number of users $u$ that have $k$ in their search history $s_u$. We call this histogram the \emph{original} histogram.
\item[3.] Delete from the histogram the pairs $(k, c_k)$ with count $c_k$ smaller than $\tau$.
\item[4.] For each pair $(k,c_k)$ in the histogram, sample a random number $\eta_k$ from the Laplace distribution
   Lap$(\lambda)$\footnote{The Laplace distribution with scale
     parameter $\lambda$ has the probability density
     function $\frac{1}{2\lambda}e^{-\frac{|x|}{\lambda}}$.}, and add $\eta_k$
   to the count $c_k$, resulting in a noisy count: $\tilde{c}_k \gets c_k + \eta_k $.
\item[5.] Delete from the histogram the pairs $(k, \tilde{c}_k)$ with noisy counts $\tilde{c}_k \leq \tau'$.
\item[6.] Publish the remaining items and their noisy counts.
\end{enumerate}
\extraspacing


To understand the purpose of the various steps one has to keep in mind the privacy guarantee we would like to achieve.
Step 1., 2.~and 4.~of the algorithm are fairly standard. It is known that adding Laplacian noise to histogram counts achieves $\epsilon$-differential privacy~\cite{dwork06Calibrating}. However, the previous section explained that these steps alone result in poor utility because for large domains many infrequent items will have high noisy counts.
To deal better with large domains we restrict the histogram to items with counts at least $\tau$ in Step 2. This restriction leaks information and thus the output after Step 4.~is not $\epsilon$-differentially private. One can show that it is not even $(\epsilon, \delta)$--probabilistic differentially private (for $\delta<1/2$). Step 5.~disguises the information leaked in Step 3. in order to achieve probabilistic differential privacy.

\begin{figure}[t]
  \begin{center}
    \resizebox{!}{6cm}{\input{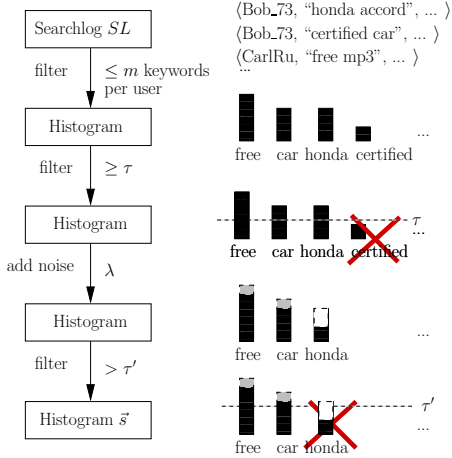}}%
  \end{center}
  \caption{ \label{fig:alg} Privacy--Preserving Algorithm.}
\end{figure}



In what follows, we will investigate the theoretical performance of ZEALOUS in terms of both privacy and utility. Section \ref{sec:alg:indist} and Section \ref{sec:alg:prob} discuss the privacy guarantees of ZEALOUS with respect to $(\epsilon, \delta)$-indistinguishability and $(\epsilon, \delta)$-probabilistic differential privacy, respectively. Section~\ref{sec:quant_comp} presents a quantitative analysis of the privacy protection offered by ZEALOUS. Sections \ref{sec:utility} and \ref{sec:separation} analyze the utility guarantees of ZEALOUS.

\subsection{Indistinguishability Analysis} \label{sec:alg:indist}

Theorem~\ref{thm:alg_san_indist} states how the parameters of ZEALOUS can be set to obtain a sanitized histogram that provides $(\epsilon',\delta')$-indistinguishability.

\extraspacing
\begin{theorem} \label{thm:alg_san_indist} \cite{KorolovaKMN09:PrivateQueries} \em
Given a search log $S$ and positive numbers $m$, $\tau$,
$\tau'$, and $\lambda$, ZEALOUS achieves $(\epsilon',\delta')$-indistinguishability, if
\begin{align}
\lambda &\geq 2m / \epsilon', \textrm{and} \label{eqn:zealous-lambda1} \\ 
\tau &= 1, \textrm{and} \label{eqn:zealous-tau1} \\
\tau' &\geq m \left(1- \frac{\log(\frac{2\delta'}{m})}{\epsilon'}\right).
\end{align}
\end{theorem}
\extraspacing

To publish not only frequent queries but also their clicks, Korolova et al.~\cite{KorolovaKMN09:PrivateQueries} suggest to first determine the frequent queries and then publish noisy counts of the clicks to their top-100 ranked documents. In particular, if we use ZEALOUS to publish frequent queries in a manner that achieves $(\epsilon',\delta')$-indistinguishability, we can also publish the noisy click distributions of the top-100 ranked documents for each of the frequent queries, by simply adding Laplacian noise to the click counts with scale $2m/\epsilon'$. Together the sanitized query and click histogram achieves $(2\epsilon', \delta')$-indistinguishability.

\subsection{Probabilistic Differential Privacy Analysis}\label{sec:alg:prob}

Given values for $\epsilon$, $\delta$, $\tau$ and $m$, the following theorem tells us how to set the parameters $\lambda$ and $\tau'$ to ensure that ZEALOUS achieves $(\epsilon, \delta)$-probabilistic differential privacy.

\extraspacing
\begin{theorem} \label{thm:alg_san} \em
Given a search log $S$ and positive numbers $m$, $\tau$,
$\tau'$, and $\lambda$, ZEALOUS achieves $(\epsilon, \delta)$-probabilistic
differential privacy, if
\begin{equation} \label{eqn:zealous-lambda}
\lambda \geq 2m / \epsilon, \textrm{
  and}
\end{equation}
\begin{equation} \label{eqn:zealous-tau}
\tau'-\tau \geq \max\left(-\lambda \ln\left(2-2e^{-\frac{1}{\lambda}}\right), -\lambda \ln\left(\frac{2\delta}{U \cdot m / \tau}\right)\right),
\end{equation}
where $U$ denotes the number of users in $S$.
\end{theorem}
\extraspacing

The proof of Theorem~\ref{thm:alg_san} can be found in Appendix~\ref{app:analysis}.

\subsection{Quantitative Comparison of Prob. Diff. Privacy and Indistinguishability for ZEALOUS}\label{sec:quant_comp}

\begin{table}[t]
\centering
\begin{tabular}{c }
\begin{minipage}[h]{8in}
{
\centering
\begin{tabular}{ c | l  l}

Privacy Guarantee & $\tau'=100$  &  $\tau'=200$\\
\hline
$\lambda=1$ ($\epsilon, \epsilon' =10$) 	&
$ \delta = 1.3\times10^{-37} $ &
$ \delta = 4.7\times10^{-81}$\\
							&
$  \delta' = 1.4\times10^{-41} $ &
$ \delta' = 5.2\times10^{-85} $\\

$\lambda=5$ ($\epsilon, \epsilon' =2$) 	&
$ \delta = 3.2\times10^{-3}$ &
$ \delta = 6.5\times10^{-12}$ \\
							&
$  \delta' = 1.4\times10^{-8} $ &
$ \delta' = 2.9\times10^{-17} $\\
\end{tabular}
}
\end{minipage}\\
\end{tabular}
\caption{ \label{table:parameters}  $(\epsilon', \delta')$-indistinguishability vs. $(\epsilon, \delta)$-probabilistic differential privacy. $U = 500k$, $m = 5$.}
\end{table}

In Table~\ref{table:parameters}, we illustrate the levels of $(\epsilon', \delta')$-indistinguishability and $(\epsilon, \delta)$-probabilistic differential privacy achieved by ZEALOUS for various noise and threshold parameters. We fix the number of users to $U = 500k$, and the maximum number of items from a user to $m = 5$, which is a typical setting that will be explored in our experiments. Table~\ref{table:parameters} shows the tradeoff between utility and privacy: A larger $\lambda$ results in a greater amount of noise in the sanitized search log (i.e., decreased data utility), but it also leads to smaller $\epsilon$ and $\epsilon'$ (i.e., stronger privacy guarantee). Similarly, when $\tau'$ increases, the sanitized search log provides less utility (since fewer items are published) but a higher level of privacy protection (as $\delta$ and $\delta'$ decreases).

Interestingly, given $\lambda$ and $\tau'$, we always have $\delta > \delta'$. This is due to the fact that $(\epsilon, \delta)$-probabilistic differential privacy is a stronger privacy guarantee than $(\epsilon', \delta')$-indistinguishability, as will be discussed in Section~\ref{sec:comp}.

\subsection{Utility Analysis} \label{sec:utility}
Next, we analyze the utility guarantee of ZEALOUS in terms of its accuracy (as defined in Section~\ref{sec:utility_def}).

\extraspacing
\begin{theorem}\label{thm:acc} \em
	Given parameters $\tau = \tau^*-\xi, \tau' = \tau^*+\xi$, noise scale $\lambda$, and a search log $S$, the inaccuracy of ZEALOUS with slack $\xi$ equals
	\begin{align*}
		  \hspace{-0.3cm}\sum_{d: f_d(S)>\tau+\xi} \hspace{-0.5cm}1/2 e^{-2\xi/\lambda}
	    + \hspace{-0.8cm}\sum_{d  \in D: f_d(S)<\tau-\xi} \hspace{-0.8cm}0
	\end{align*}
	In particular, this means that ZEALOUS is $(1-1/2 e^{-\frac{\xi}{\lambda}})$-accurate for the very-frequent items (of frequency $\geq \tau^*+\xi$) and it provides perfect accuracy for the very-infrequent items (of frequency $< \tau^*-\xi$).
\end{theorem}
\begin{proof}
	%
	%
	It is easy to see that ZEALOUS provides perfect accuracy of filtering out infrequent items. Moreover, the probability of outputting a very-frequent item is at least
	\[
		1-1/2 e^{-\frac{\xi}{\lambda}}
	\]
	which is the probability that the $\text{Lap}(\lambda)$-distributed noise that is added to the count is at least $-\xi$ so that a very-frequent item with count at least $\tau + \xi$ remains in the output of the algorithm.
	%
	This probability is at least $1/2$.
	All in all it has higher accuracy than the baseline algorithm on all  inputs with at least one very-frequent item.
\end{proof}

\subsection{Separation Result}\label{sec:separation}

Combining the analysis in Sections \ref{sec:privacy} and \ref{sec:utility}, we obtain the following separation result between $\epsilon$-differential privacy and $(\epsilon, \delta)$- probabilistic differential privacy.

\extraspacing
\begin{theorem} [Separation Result] \label{thm:separation} \em
	 Our $(\epsilon, \delta)$- probabilistic differentially private algorithm ZEALOUS is able to retain frequent items with probability at least $1/2$ while filtering out all infrequent items. On the other hand, for any $\epsilon$-differentially private algorithm that can retain frequent items with non-zero probability (independent of the input database), its inaccuracy for large item domains is larger than an algorithm that always outputs an empty set.
\end{theorem}
\extraspacing

\section{Comparing Indistinguishability with Probabilistic Differential Privacy}
\label{sec:comp} \label{sec:comp:rel} \label{sec:comp:ex}

In this section we study the relationship between $(\epsilon,\delta)$-probabilistic differential privacy and  $(\epsilon', \delta')$-indistinguishability.
First we will prove that probabilistic differential privacy implies indistinguishability. Then we will show that the converse is not true. We show that there exists an algorithm that is $(\epsilon', \delta')$-indistinguishable yet blatantly non-$\epsilon$-differentially private (and also not $(\epsilon, \delta)$-probabilistic differentially private for any value of $\epsilon$ and $\delta<1$). This fact might convince a data publisher to strongly prefer an algorithm that achieves $(\epsilon,\delta)$-probabilistic differential privacy over one that is only known to achieve  $(\epsilon', \delta')$-indistinguishability. It also might convince researchers to analyze the probabilistic privacy guarantee of algorithms that are only known to be indistinguishable  as in~\cite{DworkKMMN06_OurDataOurselves} or~\cite{NissimRS07:Smooth}.


First we show that our definition implies $(\epsilon,\delta)$-indistinguishability.

\extraspacing
\begin{proposition}\label{prop:prob_implies_ind} \em
If an algorithm $\San$ is $(\epsilon, \delta)$-probabilistic differentially private then it is also $(\epsilon,\delta)$-indistinguishable.
\end{proposition}
\extraspacing

The proof of Proposition~\ref{prop:prob_implies_ind} can be found in Appendix~\ref{app:comp_impl}. The converse of Proposition~\ref{prop:prob_implies_ind} does not hold. In particular, there exists an an algorithm that is $(\epsilon', \delta')$-indistinguishable but not $(\epsilon,\delta)$-probabilistic differentially private for any choice of $\epsilon$ and $\delta<1$, as illustrated in the following example.


\extraspacing
\begin{example}\label{ex:alg_privacy _breach}
Consider the following algorithm that takes as input a search log $S$ with search histories of $U$ users. 

\extraspacing
\noindent\textbf{Algorithm $\hat{\mathcal{A}}$} \\
\noindent\textbf{Input:}  Search log $S \in \mathcal{D}^{U}$
\begin{enumerate}
\item[1.]
Sample uniformly at random a single search history from the set of all histories excluding the first user's search history.
\item[2.]
Return this search history.
\end{enumerate}
\extraspacing 

\eat{
	\begin{algorithm}[t]
	  \caption{ $\San'$ ($S$)   \label{alg:privacy_breach}}
	  \begin{algorithmic}
		\State \Return random search history not equal to the history of the first user.
	  \end{algorithmic}
	\end{algorithm}

	\linesnumbered
	\begin{algorithm}[t!]
	\dontprintsemicolon
	  \caption{ $\San'$   \label{alg:prob_privacy}}
		\KwIn{database $S \in D^n$}
		\lFor {$i=1$ to $n-1$}
			{$S_i = \bot$\;}
		$S_0=$ sample uniform at random from $D-S_0$\;
	    \Return S\;
	\end{algorithm}
	}

	The following proposition analyzes the privacy of Algorithm $\hat{\mathcal{A}}$.

    \extraspacing
	\begin{proposition} \em
		For any finite domain of search histories $\mathcal{D}$
		Algorithm $\hat{\mathcal{A}}$ is $(\epsilon', 1/(|\mathcal{D}|-1))$-indistinguishable for all $\epsilon '>0$ on inputs from $\mathcal{D}^U$.
	\end{proposition}
    \extraspacing
    
	The proof can be found in Appendix~\ref{app:comp_alg}. The next proposition shows that every single output of the algorithm constitutes a privacy breach.

    \extraspacing
	\begin{proposition} \em
		For any search log $S$, the output of Algorithm $\hat{\mathcal{A}}$ constitutes a privacy breach according to $\epsilon$-differentially privacy for any value of $\epsilon$.
	\end{proposition}
	\begin{proof}
		Fix an input $S$ and an output $O$ that is different from the search history of the first user.
		Consider the input $S'$ differing from $S$ only in the first user history, where $S'_1 = O$.
	 	Here, 	\[ 1/(|\mathcal{D}|-1) = \Pr[\San(S) = O]  \not\leq e^{\epsilon}\Pr[\San(S') = O] = 0. \]
		Thus the output $S$ breaches the privacy of the first user according to $\epsilon$-differentially privacy.
	\end{proof}
    \extraspacing
    
	\begin{corollary} \label{prop:ind_is_weak} \em
		Algorithm $\hat{\mathcal{A}}$ is $(\epsilon', 1/(|\mathcal{D}|-1))$-indistinguishable for all $\epsilon' >0$. But it is not $(\epsilon, \delta)$-probabilistic differentially private for any $\epsilon$ and any $\delta<1$.
	\end{corollary}
\end{example}

\extraspacing



By Corollary~\ref{prop:ind_is_weak}, an algorithm that is $(\epsilon', \delta')$-indistinguishable may not achieve any form of $(\epsilon, \delta)$-probabilistic differential privacy, even if $\delta'$ is set to an extremely small value of $1/(|\mathcal{D}|-1)$. This illustrates the significant gap between $(\epsilon', \delta')$-indistinguishable and $(\epsilon, \delta)$-probabilistic differential privacy. 

\section{Choosing Parameters} \label{sec:exp_para}

\begin{figure*}[t!]
    \begin{center}
        \begin{tabular}{c c c c}  \hspace{-7mm}\epsfig{file=  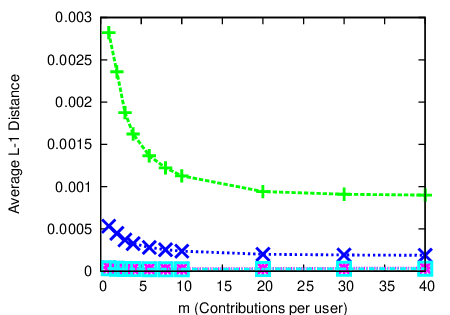,width=0.26\textwidth} & \hspace{-7mm}\epsfig{file=  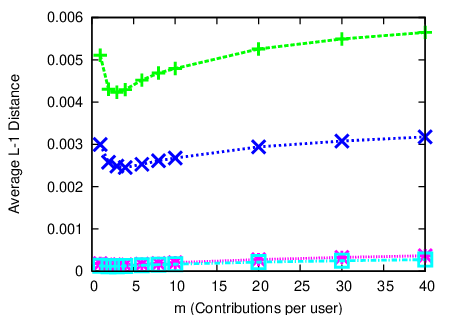,width=0.26\textwidth} & \hspace{-7mm}\epsfig{file=  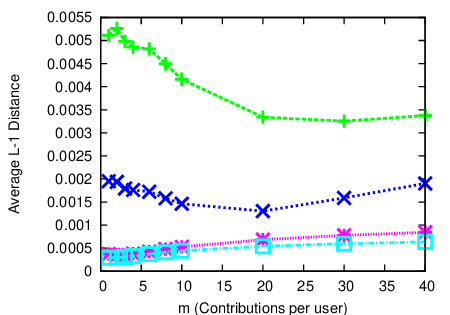,width=0.26\textwidth} & \hspace{-7mm}\epsfig{file=  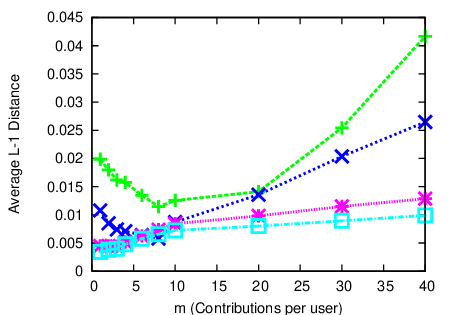,width=0.26\textwidth} \\
        \hspace{-7mm}\epsfig{file=  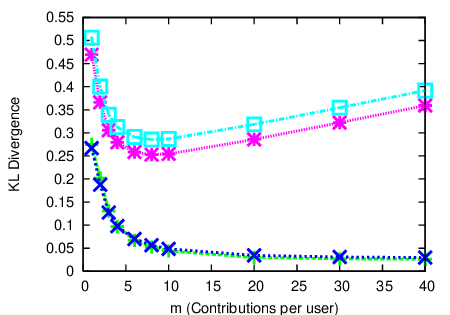,width=0.26\textwidth} & \hspace{-7mm}\epsfig{file=  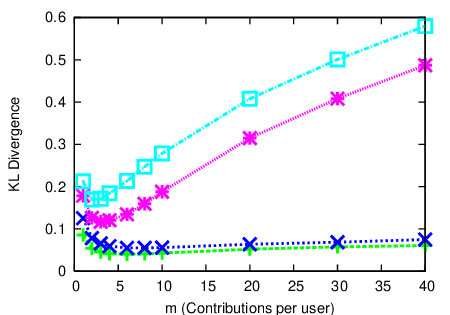,width=0.26\textwidth} & \hspace{-7mm}\epsfig{file=  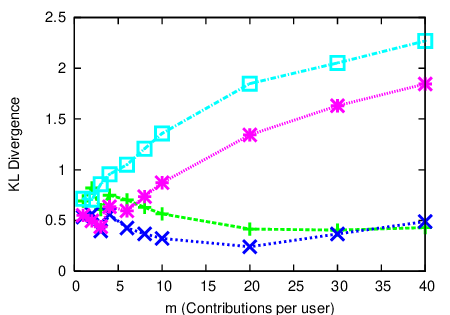,width=0.26\textwidth} & \hspace{-7mm}\epsfig{file=  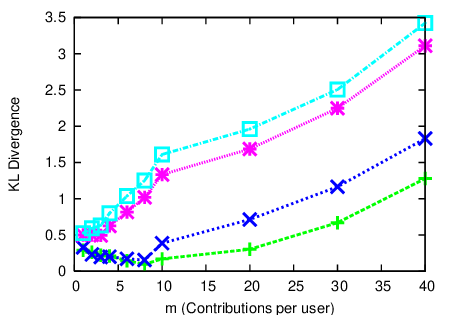,width=0.26\textwidth} \\
        \hspace{-7mm}\epsfig{file=  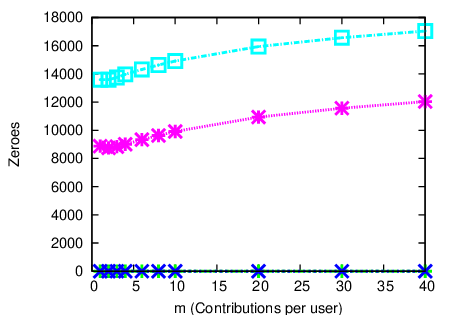,width=0.26\textwidth} & \hspace{-7mm}\epsfig{file=  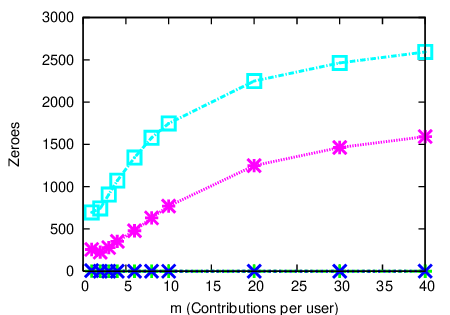,width=0.26\textwidth} & \hspace{-7mm}\epsfig{file=  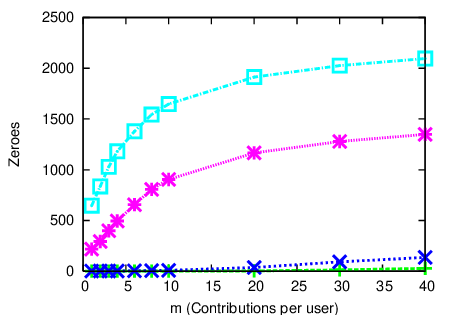,width=0.26\textwidth} & \hspace{-7mm}\epsfig{file=  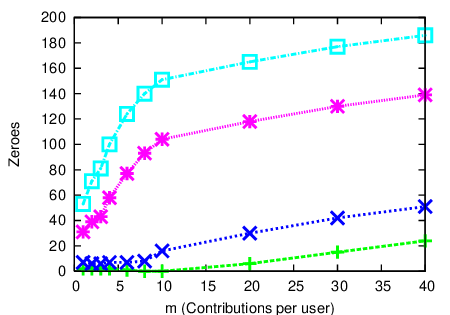,width=0.26\textwidth} \\
        \hspace{-7mm}\epsfig{file=  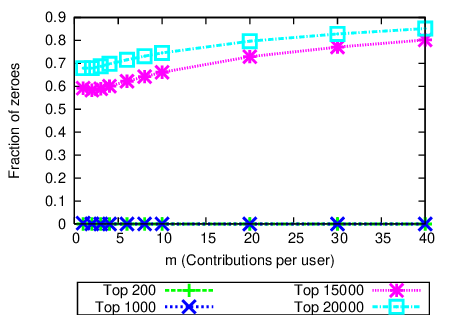,width=0.26\textwidth} & \hspace{-7mm}\epsfig{file=  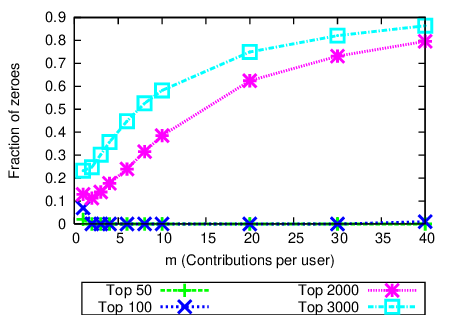,width=0.26\textwidth} & \hspace{-7mm}\epsfig{file=  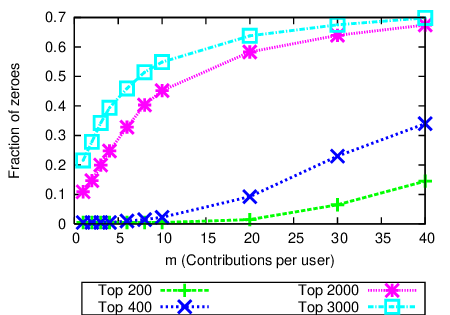,width=0.26\textwidth} & \hspace{-7mm}\epsfig{file=  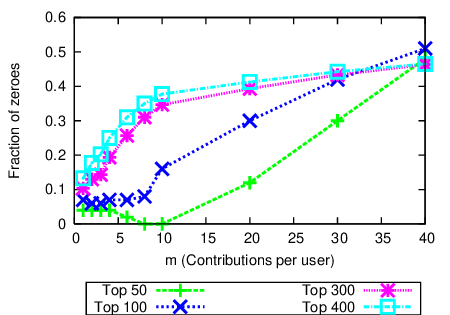,width=0.26\textwidth} \\

 (a) Keywords & (b) Queries & (c) Clicks & (d)  Query Pairs\\
\end{tabular}
\end{center}
\caption{Preservation of the counts of the top-$j$ most frequent items by ZEALOUS under varying $m$. The domain of items are keywords, queries, clicks, and query pairs. Preservation is measured as the average L1-distance and KL-divergence of the released counts and their true counts and the number and fraction of unpublished top-$j$ most frequent items  are shown.} \label{fig:stat:vary_m}
\end{figure*}

\begin{figure*}[t!]
    \begin{center}
        \begin{tabular}{c c c c}

        \hspace{-7mm}\epsfig{file=  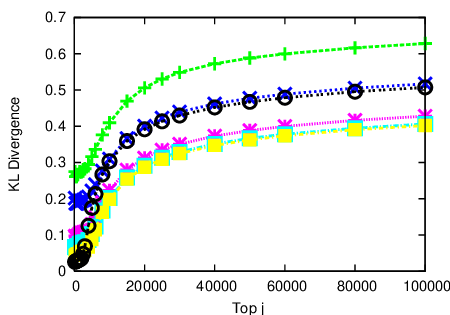,width=0.26\textwidth} & \hspace{-7mm}\epsfig{file=  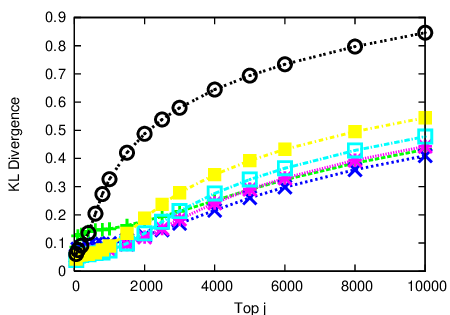,width=0.26\textwidth} & \hspace{-7mm}\epsfig{file=  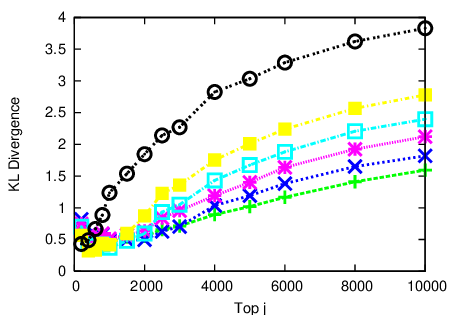,width=0.26\textwidth} & \hspace{-7mm}\epsfig{file=  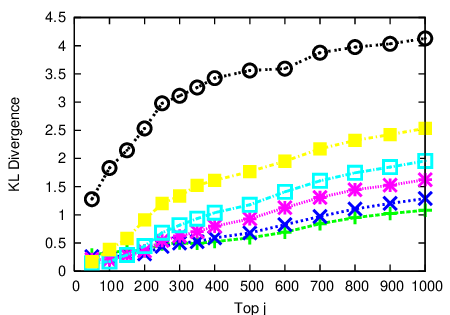,width=0.26\textwidth} \\

\hspace{-7mm}\epsfig{file=  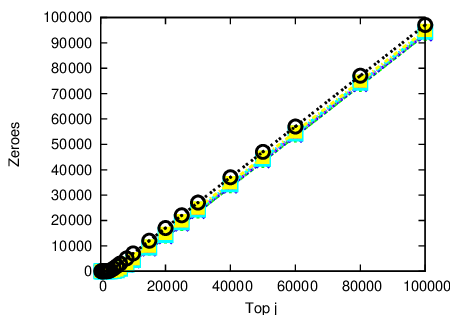,width=0.26\textwidth} & \hspace{-7mm}\epsfig{file=  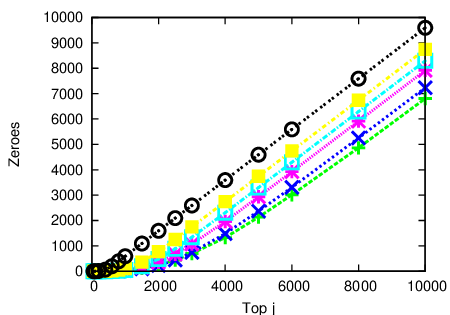,width=0.26\textwidth} & \hspace{-7mm}\epsfig{file=  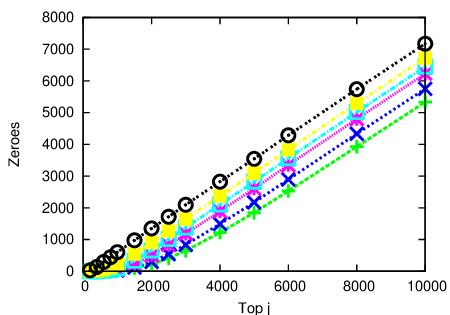,width=0.26\textwidth} & \hspace{-7mm}\epsfig{file=  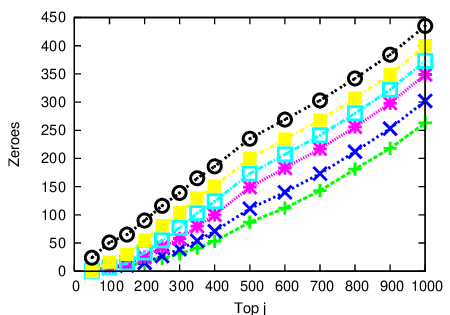,width=0.26\textwidth} \\
        \hspace{-7mm}\epsfig{file=  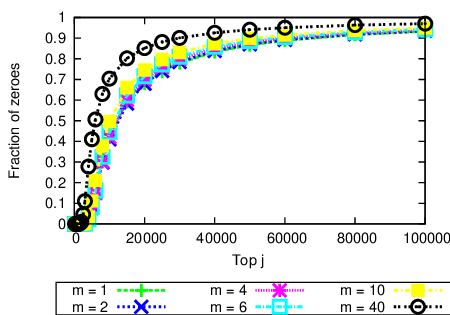,width=0.26\textwidth} & \hspace{-7mm}\epsfig{file=  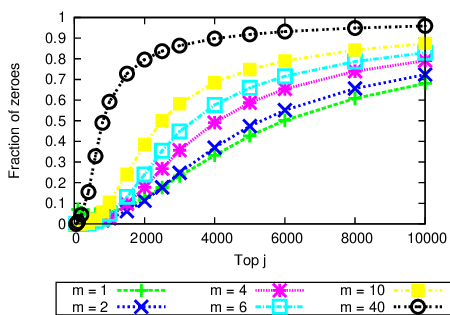,width=0.26\textwidth} & \hspace{-7mm}\epsfig{file=  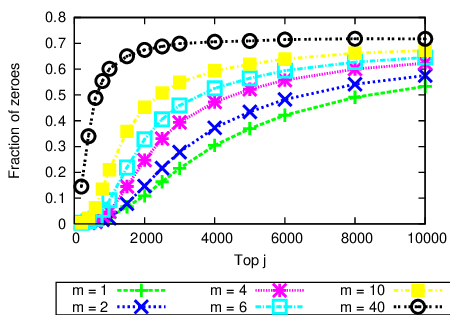,width=0.26\textwidth} & \hspace{-7mm}\epsfig{file=  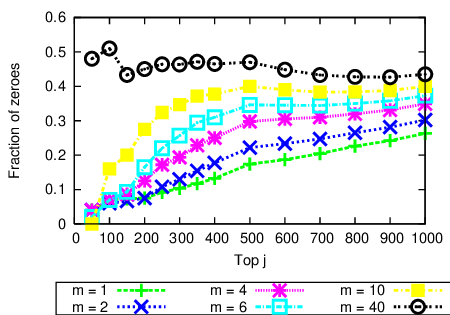,width=0.26\textwidth} \\

 (a) Keywords & (b) Queries & (c) Clicks & (d)  Query Pairs\\
\end{tabular}
\end{center}
\caption{Same as Figure~\ref{fig:stat:vary_m} except that the x-axis is now varying $j$ for top-$j$ and not $m$ for the number of contributions.} \label{fig:stat:vary_topk}
\end{figure*}

\eat{In this section we explain how to set the parameters, namely $\epsilon, \delta, m$ and $\tau$ for the ZEALOUS algorithm.
 Theorem~\ref{thm:alg_san} tells us how to set $\lambda, \tau'$ for fixed $m, \tau$ if we want to achieve $(\epsilon, \delta)$-probabilistic differential privacy.

We explore different levels of differential privacy by varying $\epsilon$ and different levels of anonymity by varying $k$. In all experiments we set $\delta=0.001$. Thus the probability that the output of ZEALOUS could breach the privacy of a user is appropriately small.
}

Apart from the privacy parameters $\epsilon$ and $\delta$, ZEALOUS
requires the data publisher to specify two more parameters: $\tau$,
the first threshold used to eliminate keywords with low counts (Step
3), and $m$, the number of contributions per user. These parameters
affect both the noise added to each count as well as the second
threshold $\tau'$. Before we discuss the choice of these parameters
we explain the general setup of our experiments.

\vspace{2mm} \noindent\textbf{Data.}
In our experiments we work with a search log of user queries from 
a major search engine collected from 500,000 users over a period of
one month.
This search log contains about one million distinct keywords, three
million distinct queries, three million distinct query pairs,
and 4.5 million distinct clicks.

\vspace{2mm} \noindent\textbf{Privacy Parameters.} In all experiments
we set $\delta=0.001$. Thus the probability that the output of ZEALOUS
could breach the privacy of any user is very small. We
explore different levels of $(\epsilon, \delta)$-probabilistic
differential privacy by varying $\epsilon$.

\subsection{Choosing Threshold $\tau$}

\begin{table}[t!]
\centering
\begin{tabular}{l | r r r r r r}
    $\tau$ & 1 &  3  & 4 & 5  & 7  & 9  \\
    \hline
    $\tau'$ & 81.1 & 78.7 & 78.6 & 78.7 &  79.3  & 80.3\\
\end{tabular}
\caption{ $\tau'$ as a function of $\tau$ for $m = 2$, $\epsilon = 1$, $\delta = 0.01$ \label{table:tau'} }
\end{table}

\eat{
\begin{table*}[t!]
\centering
\begin{tabular}{l | r r r r r  r r r r r}
    $\tau$ & 1 & 2 & 3 & 4 & 5 & 6 & 7 & 8 & 9 & 10 \\
    \hline
    $\tau'$ & 81.1205 & 79.3479 & 78.7260 & 78.5753 & 78.6827 & 78.9534 & 79.3368 & 79.8027 & 80.3316 & 80.9101 \\
\end{tabular}
\caption{ $\tau'$ as a function of $\tau$ for $m = 2$, $\epsilon = 1$, $\delta = 0.01$ \label{table:tau'} }
\end{table*}
} We would like to retain as much information as possible in the
published search log. A smaller value for $\tau'$ immediately leads to
a histogram with higher utility because fewer items and their noisy
counts are filtered out in the last step of ZEALOUS. Thus if we choose
$\tau$ in a way that minimizes $\tau'$ we maximize the utility of the
resulting histogram. Interestingly, choosing $\tau = 1$ does not
necessarily minimize the value of $\tau'$. Table~\ref{table:tau'}
presents the value of $\tau'$ for different values of $\tau$ for $m =
2$ and $\epsilon = 1$.  Table~\ref{table:tau'} shows that for our
parameter settings $\tau'$ is minimized when $\tau = 4$.  We can show the
following optimality result which tells us how to choose $\tau$
optimally in order to maximize utility.
\begin{proposition} \label{prop:choose_tau} \em For a fixed $\epsilon,
  \delta$ and $m$ choosing $\tau=\lceil 2m/\epsilon \rceil$ minimizes
  the value of $\tau'$.
\end{proposition}
The proof follows from taking the derivative of $\tau'$ as a function
of $\tau$ (based on Equation (\ref{eqn:zealous-tau})) to determine its
minimum.

\subsection{Choosing the Number of Contributions $m$}

Proposition~\ref{prop:choose_tau} tells us how to set $\tau$ in order
to maximize utility. Next we will discuss how to set $m$ optimally. We
will do so by studying the effect of varying $m$ on the coverage and
the precision of the top-$j$ most frequent items in the sanitized
histogram.  The \emph{top-j coverage of a sanitized search log} is
defined as the fraction of distinct items among the top-$j$ most
frequent items in the original search log that also appear in the
sanitized search log. The \emph{top-j precision of a sanitized search
  log} is defined as the distance between the relative frequencies in
the original search log versus the sanitized search log for the
top-$j$ most frequent items. In particular, we study two distance
metrics between the relative frequencies: the average L-1 distance and
the KL-divergence.

\begin{table}[t]
		\centering
(A)	Number of distinct items released under ZEALOUS. 
\\
\begin{tabular}{l | r r r r r}
        $m$     &   1       &  4   &   8   &  20  &  40\\
        \hline
    keywords    &   6667    & 6043  &   5372  &  4062 & 2964\\
    queries     &   3334    & 2087  &   1440  &  751 & 408\\
    clicks      &   2813    & 1576  &   1001  &  486 & 246\\
    query pairs &   331     & 169   &   100   &  40 & 13
\end{tabular}\\
\vspace{5mm}

(B)  Number of total items $\times 10^3$ released under ZEALOUS.
\\
\begin{tabular}{l | r r r r r}
        $m$     &   1               &   4     &   8   & 20 & 40\\
        \hline
    keywords    &   329        & 1157       &   1894 &  3106 & 3871 \\
    queries     &   147        & 314         & 402 &  464 & 439 \\
    clicks      &   118        & 234       & 286 & 317 & 290 \\
    query pairs & 8         & 14     & 15 & 12 & 7
\end{tabular}
\caption{
\label{table:total_items}
\label{table:distinct_items}  }
\end{table}
 
As a first study of coverage, Table~\ref{table:distinct_items}
shows the number of distinct items (recall that items can be keywords,
queries, query pairs, or clicks) in the sanitized search log as $m$
increases.  We observe that coverage decreases as we increase
$m$. Moreover, the decrease in the number of published items is more
dramatic for larger domains than for smaller domains. The number of
distinct keywords decreases by $55\%$ while at the same time the
number of distinct query pairs decreases by $96\%$ as we increase $m$
from $1$ to $40$.  This trend has two reasons. First, from
Theorem~\ref{thm:alg_san} and Proposition~\ref{prop:choose_tau} we see
that threshold $\tau'$ increases super-linearly in $m$. Second, as $m$
increases the number of keywords contributed by the users increases
only sub-linearly in $m$; fewer users are able to supply $m$ items for
increasing values of $m$. Hence, fewer items pass the threshold
$\tau'$ as $m$ increases. The reduction is larger for query pairs than
for keywords, because the average number of query pairs per user is
smaller than the average number of keywords per user in the original
search log (shown in Table~\ref{table:avg-m}).

\begin{table}[t]
	\centering
\begin{tabular}{ r r r r}
             keywords & queries & clicks & query pairs\\
            \hline
   			 56 & 20 & 14 & 7
\end{tabular}
\caption{ Avg number of items per user in the original search log
\label{table:avg-m} }
\end{table}

To understand how $m$ affects precision, we measure the total sum of
the counts in the sanitized histogram as we increase $m$ in
Table~\ref{table:total_items}. Higher total counts offer the
possibility to match the original distribution at a finer 
granularity.  We
observe that as we increase $m$, the total counts increase until a
tipping point is reached after which they start decreasing again. This
effect is as expected for the following reason: As $m$ increases, each
user contributes more items, which leads to higher counts in the
sanitized histogram. However, the total count increases only
sub-linearly with $m$ (and even decreases) due to the reduction in
coverage we discussed above. We found that the tipping point where the
total count starts to decrease corresponds approximately to the
average number of items contributed by each user in the original
search log (shown in Table~\ref{table:avg-m}). This suggests that we
should choose $m$ to be smaller than the average number of items,
because it offers better coverage, higher total counts, and reduces the
noise compared to higher values of $m$.

Let us take a closer look at the precision and coverage of the
histograms of the various domains in Figures~\ref{fig:stat:vary_m} and
\ref{fig:stat:vary_topk}. In Figure~\ref{fig:stat:vary_m} we vary $m$
between 1 and 40. Each curve plots the precision or coverage of the
sanitized search log at various values of the top-$j$ parameter in
comparison to the original search log.  We vary the top-$j$ parameter
but never choose it higher than the number of distinct items in the
original search log for the various domains. The first two rows plot
precision curves for the average L-1 distance (first row) and the
KL-divergence (second row) of the relative frequencies.  
The lower two
rows plot the coverage curves, i.e., the total number of top-$j$ items
(third row) and the relative number of top-$j$ items (fourth row) in
the original search log that do not appear in sanitized search log.
First, observe that the coverage decreases as $m$ increases, which
confirms our discussion about the number of distinct items. Moreover,
we see that the coverage gets worse for increasing values of the
top-$j$ parameter. This illustrates that ZEALOUS gives better utility
for the more frequent items. Second, note that for small values of the
top-$j$ parameter, values of $m > 1$ give better precision. However,
when the top-$j$ parameter is increased, $m=1$ gives better precision
because the precision of the top-$j$ values degrades due to items no
longer appearing in the sanitized search log due to the increased
cutoffs.

Figure~\ref{fig:stat:vary_topk} shows the same statistics varying the
top-$j$ parameter on the x-axis. Each curve plots the precision for $m
= 1, 2, 4, 8, 10, 40$, respectively. Note that $m=1$ does not always
give the best precision; for keywords, $m = 8$ has the lowest
KL-divergence, and for queries, $m = 2$ has the lowest KL-divergence.
As we can see from these results, there are two ``regimes'' for
setting the value of $m$. If we are mainly interested in coverage,
then $m$ should be set to $1$. However, if we are only interested in a
few top-$j$ items then we can increase precision by choosing a larger
value for $m$; and in this case we recommend the average number of
items per user.

We will see this dichotomy again in our real applications of search
log analysis: The index caching application does not require high
coverage because of its storage restriction. However, high precision
of the top-$j$ most frequent items is necessary to determine which of
them to keep in memory. On the other hand, in order to generate many
query substitutions, a larger number of distinct queries and query
pairs is required. Thus $m$ should be set to a large value for index
caching and to a small value for query substitution.

\section{Application-Oriented Evaluation}
 \label{sec:apps}  \label{sec:exp_comp}

In this section we show the results of an application-oriented
evaluation of privacy-preserving search logs generated by ZEALOUS in comparison to a $k$-anonymous search log and
the original search log as points of comparison.
Note that our utility
evaluation does not determine the ``better'' algorithm since when
choosing an algorithm in practice one has to consider both the utility
and disclosure limitation guarantees of an algorithm.  
Our results show the ``price'' that we have to pay (in terms of
decreased utility) when we give the stronger guarantees of
(probabilistic versions of) differential privacy as opposed to
$k$-anonymity.

%

\vspace{2mm} \noindent\textbf{Algorithms.}

We experimentally compare the utility of ZEALOUS against a
representative $k$-anonymity algorithm by Adar for publishing search
logs~\cite{Adar07:AnonQueryLogs}.  Recall that Adar's Algorithm
creates a $k$-query anonymous search log as follows: First all queries
that are posed by fewer than $k$ distinct users are eliminated. Then
histograms of keywords, queries, and query pairs from the $k$-query
anonymous search log are computed.
%
ZEALOUS can be used to achieve
$(\epsilon',\delta')$-indistinguishability as well as
$(\epsilon,\delta)$-probabilistic differential privacy. For the ease
of presentation we only show results with probabilistic differential
privacy; using Theorems~\ref{thm:alg_san_indist} and
\ref{thm:alg_san} it is straightforward to compute the corresponding
indistinguishability guarantee.  For brevity, we refer to the
$(\epsilon,\delta)$-probabilistic differentially private algorithm as
$\epsilon$--Differential in the figures.

%

\vspace{2mm} \noindent\textbf{Evaluation Metrics.} 

We evaluate the performance of the algorithms in two ways.  First, we
measure how well the output of the algorithms preserves selected
statistics of the original search log.  Second, we pick two real
applications from the information retrieval community to evaluate the
utility of ZEALOUS: Index caching as a representative application for
search performance, and query substitution as a representative
application for search quality. Evaluating the output of ZEALOUS with
these two applications will help us to fully understand the
performance of ZEALOUS in an application context.  We first describe
our utility evaluation with statistics in Section
\ref{sec:exp_statistics} and then our evaluation with real
applications in Sections \ref{sec:exp_index} and \ref{sec:exp_subs}.

\subsection{General Statistics} \label{sec:exp_statistics}

We explore different statistics that measure the difference of
sanitized histograms to the histograms computed using the original search log. We
analyze the histograms of keywords, queries, and query pairs for
both sanitization methods. For clicks we only consider ZEALOUS histograms 
since a $k$-query anonymous search log is not designed to publish
click data.

In our first experiment we compare the distribution of the counts in
the histograms. Note that a $k$-query anonymous search log will never
have query and keyword counts below $k$, and similarly a ZEALOUS histogram 
will never have counts below $\tau'$.  We choose
$\epsilon = 5, m = 1$ for which threshold $\tau' \approx
10$. Therefore we deliberately set $k=10$ such that $k \approx \tau'$.

\begin{figure*}[t!]
    \begin{center}
    \begin{tabular}{c c c c}
    \hspace{-5mm}\epsfig{file=  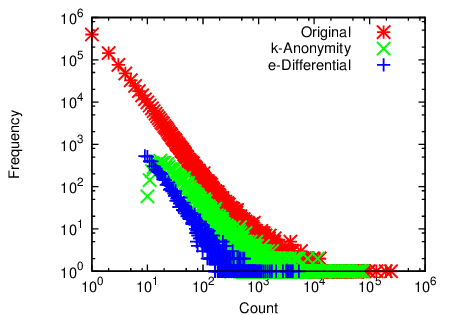,width=0.25\textwidth}    &
    \hspace{-5mm}\epsfig{file=  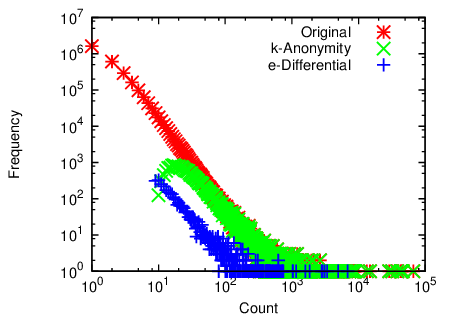,width=0.25\textwidth}      &

    \hspace{-5mm}\epsfig{file=  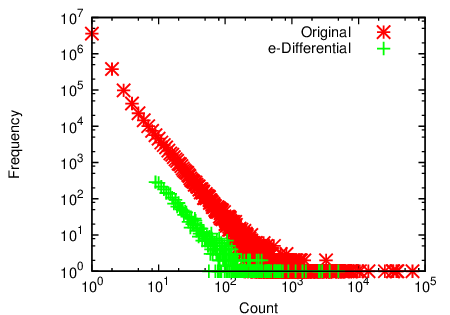,width=0.25\textwidth} &
    \hspace{-5mm}\epsfig{file=  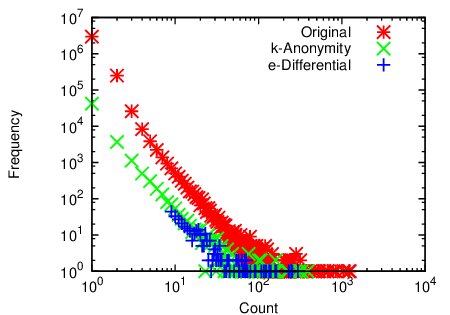,width=0.25\textwidth}\\
        (a) Keyword Counts & (b) Query Counts &     (c) Click Counts & (d) Query Pair Counts
        \end{tabular}
        \end{center}
        \caption{ Distributions of counts in the histograms over keywords, queries, clicks, and query pairs in the original search log and its sanitized versions created by $10$-anonymity and $5$-probabilistic differential privacy (with $m=1$). 
} \label{fig:distr}
\end{figure*}

\begin{figure*}[t!]
    \begin{center}
        \begin{tabular}{ c c c c} \hspace{-4mm}\epsfig{file=  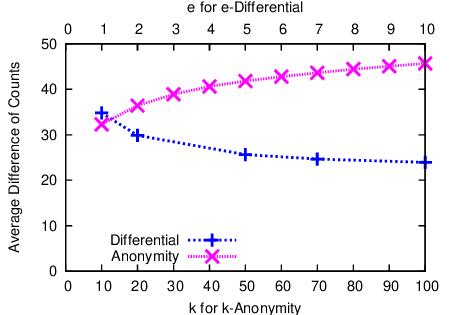,width=0.26\textwidth} & \hspace{-7mm}\epsfig{file=  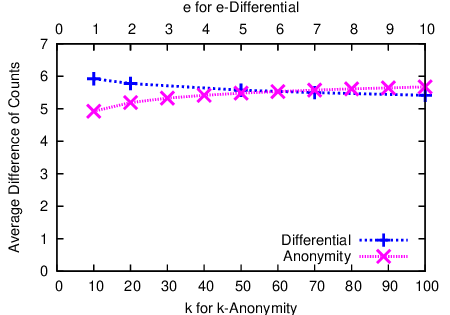,width=0.26\textwidth} & \hspace{-7mm}\epsfig{file=  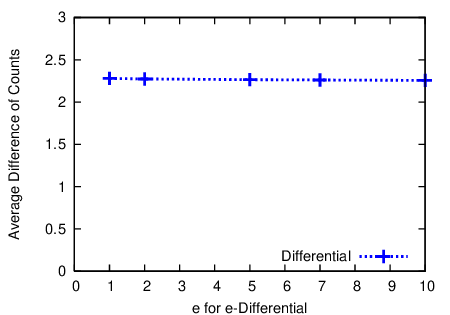,width=0.26\textwidth} &
    \hspace{-7mm}\epsfig{file=  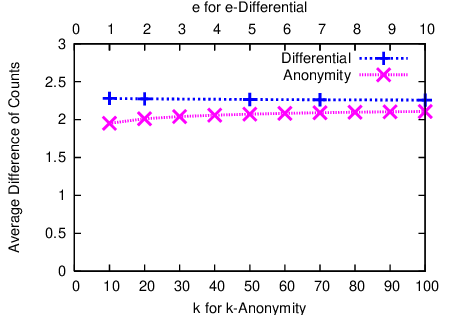,width=0.26\textwidth} \\
 (a) Keywords & (b) Queries & (c) Clicks & (d)  Query Pairs \\
\end{tabular}
\end{center}
\caption{ Average difference between counts of items in the original
histogram and the $\epsilon$-probabilistic differential privacy-preserving
histogram, and the $k$-anonymous histogram for varying parameters $\epsilon$ (with $m=1$) and $k$. 
} \label{fig:stats}
\end{figure*}

Figure~\ref{fig:distr} shows the distribution of the counts in the
histograms on a log-log scale. Recall that the $k$-query anonymous
search log does not contain any click data, and thus it does not
appear in Figure~\ref{fig:distr}(c). We see that the power-law shape
of the distribution is well preserved. However, the total frequencies
are lower for the sanitized search logs than the frequencies in the
original histogram because the algorithms filter out a large number of
items. We also see the cutoffs created by $k$ and $\tau'$.  We observe
that as the domain increases from keywords to clicks and query pairs,
the number of items that are not frequent in the original search log
increases. For example, the number of clicks with count equal to one
is an order of magnitude larger than the number of keywords with count
equal to one.

While the shape of the count distribution is well preserved, we would
also like to know whether the counts of frequent keywords, queries,
query pairs, and clicks are also preserved and what impact the privacy
parameters $\epsilon$ and the anonymity parameter $k$ have.
Figure~\ref{fig:stats} shows the average differences to the counts in
the original histogram. We scaled up the counts in sanitized
histograms by a common factor so that the total counts were equal to
the total counts of the original histogram, then we calculated the
average difference between the counts. The average is taken over all
keywords that have non-zero count in the original search log. As such
this metric takes both coverage and precision into account.

As expected, with increasing $\epsilon$ the average difference
decreases, since the noise added to each count decreases. Similarly,
by decreasing $k$ the accuracy increases because more queries will
pass the threshold. Figure~\ref{fig:stats} shows that the average
difference is comparable for the $k$--anonymous histogram and the
output of ZEALOUS. Note that the output of ZEALOUS for keywords is
more accurate than a $k$-anonymous histogram for all values of
$\epsilon > 2$.  For queries we obtain roughly the same average
difference for $k = 60$ and $\epsilon = 6$. For query pairs the
$k$-query anonymous histogram provides better utility.

We also computed other metrics such as the root-mean-square value of
the differences and the total variation difference; they all reveal
similar qualitative trends. Despite the fact that ZEALOUS
disregards many search log records (by throwing out all but $m$
contributions per user and by throwing out low frequent counts),
ZEALOUS is able to preserve the overall distribution well.

\subsection{Index Caching} \label{sub-sec:mem} \label{sec:exp_index}


In the index caching problem, we aim to cache in-memory a set 
of posting lists that maximizes the hit-probability over all keywords
(see Section\ref{sec:expUtilM}).
In our experiments, we use an improved version of the algorithm
developed by Baeza--Yates to decide which posting lists
should be kept in memory~\cite{Baeza04}. Our algorithm first assigns each keyword a
score, which equals its frequency in the search log divided by the
number of documents that contain the keyword. Keywords are chosen
using a greedy bin-packing strategy where we sequentially add posting lists from
the keywords with the highest score until the memory is filled.
In our experiments we
fixed the memory size to be 1 GB, and each document posting to be 8
Bytes (other parameters give comparable results). Our inverted index
stores the document posting list for each keyword sorted according to
their relevance which allows to retrieve the documents in the order of
their relevance. We truncate this list in memory to contain at most
200,000 documents. Hence, for an incoming query the search engine
retrieves the posting list for each keyword in the query either from
memory or from disk.  If the intersection of the posting lists happens
to be empty, then less relevant documents are retrieved from disk for
those keywords for which only the truncated posting list is kept on
memory.
\begin{figure}[t!]
    \begin{center}
        \begin{tabular}{ c c} 
\hspace{-7mm}\epsfig{file=  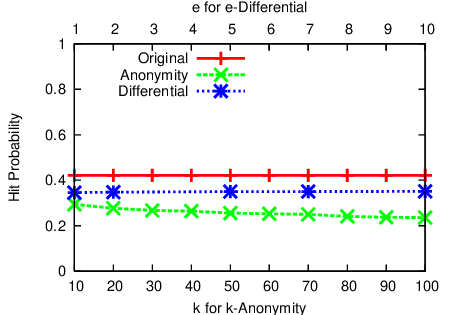,width=0.25\textwidth} & \hspace{-5mm} \epsfig{file=  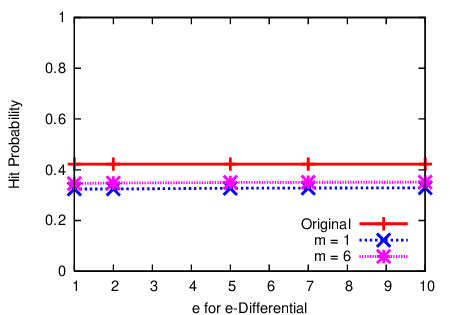,width=0.25\textwidth} \\
 (a)   & (b) \\
\end{tabular}
\end{center} \caption{Hit probabilities of the inverted index construction based on the original search log, the $k$-anonymous search log, and the ZEALOUS histogram under varying parameters $k$ and $\epsilon$ (and contributions $m$ in (b)).
} \label{fig:memory}
\end{figure}

Figure~\ref{fig:memory}(a) shows the hit--probabilities of the
inverted index constructed using the original search log, the
$k$-anonymous search log, and the ZEALOUS histogram (for $m=6$) with
our greedy approximation algorithm. We observe that our ZEALOUS
histogram achieves better utility than the $k$-query anonymous search
log for a range of parameters. We note that the utility suffers only
marginally when increasing the privacy parameter or the anonymity
parameter (at least in the range that we have considered). This can be
explained by the fact that it requires only a few very frequent
keywords to achieve a high hit--probability. Keywords with a big
positive impact on the hit-probability are less likely to be filtered
out by ZEALOUS than keywords with a small positive impact. This
explains the marginal decrease in utility for increased privacy.

As a last experiment we study the effect of varying $m$ on the
hit-probability in Figure~\ref{fig:memory}(b). We observe that the hit
probability for $m = 6$ is above 0.36 whereas the hit probability for
$m = 1$ is less than 0.33. As discussed a higher value for $m$
increases the accuracy, but reduces the coverage. Index caching really
requires roughly the top 85 most frequent keywords that are still
covered when setting $m=6$. We also experimented with higher values of
$m$ and observed that the hit-probability decreases at some point.

\subsection{Query Substitution} \label{sub-sec:subst} \label{sec:exp_subs}

Algorithms for query substitution examine query pairs to learn how users re-phrase queries. We use an algorithm developed by Jones et al. in which  related
queries for a query are identified in two steps~\cite{JonesRMG06:QuerySubstitution}. First, the query is
partitioned into subsets of keywords, called {\em phrases}, based on
their mutual information. Next, for each phrase, candidate query
substitutions are determined based on the distribution of 
queries.

We run this algorithm to generate ranked substitution on the sanitized
search logs.  We then compare these rankings with the rankings
produced by the original search log which serve as ground truth.  To
measure the quality of the query substitutions, we compute the
precision/recall, MAP (mean average precision) and NDG (normalized
discounted cumulative gain) of the top-$j$ suggestions for each query;
let us define these metrics next.

Consider a query $q$ and its list of top-$j$ ranked substitutions
$q_0', \dots, q_{j-1}'$ computed based on a sanitized search log. We
compare this ranking against the top-$j$ ranked substitutions $q_0,
\dots, q_{j-1}$ computed based on the original search log as
follows. The \emph{precision} of a query $q$ is the fraction of substitutions from
the sanitized search log that are also contained in our ground truth
ranking:
\[
\text{Precision}(q) = \frac{| \{q_0, \dots, q_{j-1}\} \cap \{q_0',
\dots, q_{j-1}'\} |}{|\{ q_0', \dots, q_{j-1}'\}|}
\]
Note, that the number of items in the ranking for a query $q$ can be
less than $j$. The \emph{recall} of a query $q$  is the fraction of substitutions in
our ground truth that are contained in the substitutions from the
sanitized search log:
\[
\text{Recall}(q) = \frac{| \{q_0, \dots, q_{j-1}\} \cap \{q_0',
\dots, q_{j-1}'\} |}{|\{ q_0, \dots, q_{j-1}\}|}
\]
MAP measures the precision of the ranked items for a
query as the ratio of true rank and assigned rank:
\[
\text{MAP}(q) = \sum_{i=0}^{j-1}\frac{i+1}{\text{rank of } q_i
\text{ in } [q_0', \dots, q_{j-1}']+1},
\]
where the rank of $q_i$ is zero in case it does is not contained in
the list $[q_0', \dots, q_{j-1}']$ otherwise it is $i'$, s.t. $q_i =
q_{i'}'$.

Our last metric called NDCG measures how the relevant substitutions
are placed in the ranking list. It does not only compare the ranks of
a substitution in the two rankings, but is also penalizes highly
relevant substitutions according to $[q_0, \dots, q_{j-1}]$ that have
a very low rank in $[q_0', \dots, q_{j-1}']$. Moreover, it takes the
length of the actual lists into consideration. We refer the reader to
the paper by Chakrabarti et al.~\cite{ChakrabartiKSB08} for details on
NDCG.

The discussed metrics compare rankings for one query. To compare the
utility of our algorithms, we average over all queries. For coverage
we average over all queries for which the original search log produces
substitutions. For all other metrics that try to capture the precision
of a ranking, we average only over the queries for which the sanitized
search logs produce substitutions. We generated query substitution
only for the 100,000 most frequent queries of the original search log
since the substitution algorithm only works well given enough
information about a query.

\begin{figure*}[t]
    \begin{center}
        \begin{tabular}{ c c c c}
    \hspace{-7mm}\epsfig{file=  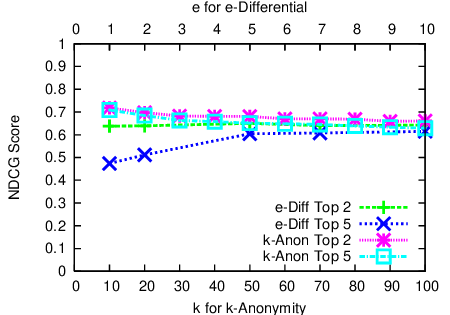, width=0.26\textwidth}&
    \hspace{-6mm}\epsfig{file=  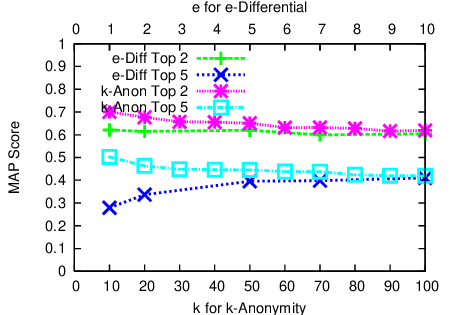,width=0.26\textwidth}&
    \hspace{-6mm}\epsfig{file=  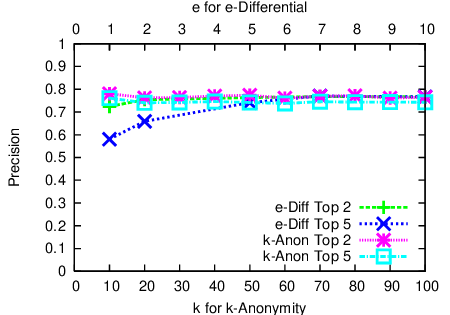, width=0.26\textwidth}   &
    \hspace{-6mm}\epsfig{file=  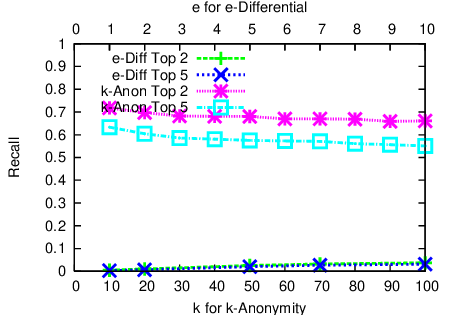, width=0.26\textwidth} \\
        (a) NDCG & (b) MAP & (c) Precision & (d) Recall
        \end{tabular}
    \end{center}
    \caption{ Quality of the query substitutions of the $\epsilon$-probabilistic differential privacy-preserving	histogram, and the $k$-anonymous histogram for varying parameters $\epsilon$ (with $m=1$) and $k$. The quality is measured by comparing the top-$2$ and top-$5$ suggested query substitutions to the ground truth recording the NDCG, MAP, precision, and recall.} \label{fig:sub}
\end{figure*}

\begin{figure}[t!]
    \begin{center}
        \begin{tabular}{ c }
            \epsfig{file=  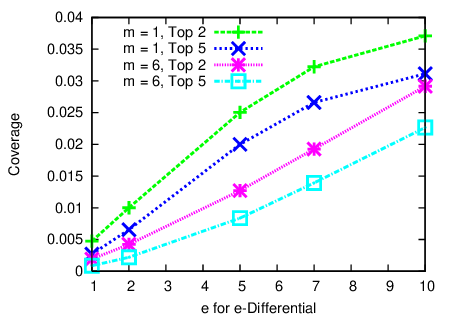,width=0.36\textwidth}
\end{tabular}
\end{center} \caption{Coverage of the query substitutions of the privacy-preserving	histogram for $m=1$ and $m=6$.} \label{fig:sub_vary_m}
\end{figure}

In Figure~\ref{fig:sub} we vary $k$ and $\epsilon$ for $m=1$ and we
draw the utility curves for top-$j$ for $j=2$ and $j=5$. We observe
that varying $\epsilon$ and $k$ has hardly any influence on 
performance. On all precision measures, ZEALOUS provides utility
comparable to $k$-query-anonymity. However, the coverage provided by
ZEALOUS is not good. This is because the computation of query
substitutions relies not only on the frequent query pairs but also on
the count of phrase pairs which record for two sets of keywords how
often a query containing the first set was followed by another query
containing the second set. Thus a phrase pair can have a high
frequency even though all query pairs it is contained in have very low
frequency. ZEALOUS filters out these low frequency query pairs and
thus loses many frequent phrase pairs.

As a last experiment, we study the effect of increasing $m$ for query
substitutions. Figure~\ref{fig:sub_vary_m} plots the average coverage
of the top-2 and top-5 substitutions produced by ZEALOUS for $m = 1$
and $m=6$ for various values of $\epsilon$. It is clear that across
the board larger values of $m$ lead to smaller coverage, thus
confirming our intuition outlined the previous section.

\section{Related Work}
Related work on anonymity in search logs~\cite{Adar07:AnonQueryLogs,HeN09:HierarchicalAnon,MotwaniN:searchlogs,YuanAnonymization09} is discussed in Section~\ref{sec:anon}. 

More recently, there has been work on privacy in search logs. Korolova et al.~\cite{KorolovaKMN09:PrivateQueries}  proposes the same basic algorithm that we propose in~\cite{GoetzMWXG09:searchlogsV2} and review in Section~\ref{sec:alg}.\footnote{In order to improve utility of the algorithm as stated in \cite{KorolovaKMN09:PrivateQueries}, we suggest to first filter out infrequent keywords using the 2-threshold approach of ZEALOUS and then publish noise counts of queries consisting of up to 3 frequent keywords and the clicks of their top ranked documents.}
They show $(\epsilon', \delta')$-indistinguishability of the algorithm whereas we show $(\epsilon, \delta)$-probabilistic differential privacy of the algorithm which is a strictly stronger guarantee (see  Section~\ref{sec:comp:rel}).
One difference is that our algorithm has two thresholds $\tau,\tau'$ as opposed to one and we explain how to set threshold $\tau$ optimally. Korolova et al.~\cite{KorolovaKMN09:PrivateQueries} set $\tau=1$ (which is not the optimal choice in many cases).
Our experiments augment and extend the experiments of Korolova et al.~\cite{KorolovaKMN09:PrivateQueries}.
We illustrate the tradeoff of setting  the number of contributions $m$ for various domains and statistics including L1-distance and KL divergence which extends \cite{KorolovaKMN09:PrivateQueries} greatly.
Our application oriented evaluation considers different applications. We compare the performance of ZEALOUS to that of $k$-query anonymity and observe that the loss in utility is comparable for anonymity and privacy while  anonymity offers a much weaker guarantee.

\section{Beyond Search Logs}

While the main focus of this paper are search logs, our results apply
to other scenarios as well.  For example, consider a retailer who
collects customer transactions. Each transaction consists of a basket
of products together with their prices, and a time-stamp. In this case
ZEALOUS can be applied to publish frequently purchased products or
sets of products. This information can also be used in a recommender system
or in a market basket analysis to decide on the goods and promotions
in a store~\cite{cite:book/DataMining/2000}.
Another example concerns monitoring the health of patients. Each time a patient sees a doctor the doctor records the diseases of the patient and the suggested treatment. It would be interesting to publish frequent combinations of diseases.

All of our results apply to the more general problem of publishing
frequent items / itemsets / consecutive itemsets.  Existing work on
publishing frequent itemsets often only tries to achieve anonymity or
makes strong assumptions about the background knowledge of an
attacker, see for example some of the references in the survey by Luo
et al.~\cite{Luo:SurveyItemsets09}.

\section{Conclusions}\label{sec:conclusions}
This paper contains a comparative study about publishing frequent keywords, queries, and clicks in search logs. We compare the disclosure limitation guarantees and the theoretical and practical utility of various approaches. Our comparison includes earlier work on anonymity and $(\epsilon',\delta')$--indistinguishability and our proposed solution to achieve  $(\epsilon, \delta)$-probabilistic differential privacy in search logs.
In our comparison, we revealed interesting relationships between indistinguishability and probabilistic differential privacy which might be of independent interest. Our results (positive as well as negative) can be applied more generally to the problem of publishing frequent items or itemsets.

A topic of future work is the development of algorithms to release useful information about \emph{infrequent} keywords, queries, and clicks in a search log while preserving user privacy.

\vspace{2mm} 
\noindent {\bf Acknowledgments.} We would like to thank our colleagues Filip
Radlinski and Yisong Yue for helpful discussions about
the usage of search logs.


\bibliographystyle{plain}
\bibliography{references}

\begin{thebibliography}{10}

\bibitem{Adar07:AnonQueryLogs}
Eytan Adar.
\newblock User 4xxxxx9: Anonymizing query logs.
\newblock In {\em WWW Workshop on Query Log Analysis}, 2007.

\bibitem{Baeza04}
Roberto Baeza-Yates.
\newblock Web usage mining in search engines.
\newblock {\em Web Mining: Applications and Techniques}, 2004.

\bibitem{barakCDKMT07:holistic}
B.~Barak, K.~Chaudhuri, C.~Dwork, S.~Kale, F.~McSherry, and K.~Talwar.
\newblock Privacy, accuracy and consistency too: A holistic solution to
  contingency table release.
\newblock In {\em PODS}, 2007.

\bibitem{NYTimes:AOL}
Michael Barbaro and Tom Zeller.
\newblock A face is exposed for aol searcher no. 4417749.
\newblock New York Times
  \url{http://www.nytimes.com/2006/08/09/technology/09aol.html?ex=1312776000en%
=f6f61949c6da4d38ei=5090}, 2006.

\bibitem{BlumLR08}
Avrim Blum, Katrina Ligett, and Aaron Roth.
\newblock A learning theory approach to non-interactive database privacy.
\newblock In {\em STOC}, pages 609--618, 2008.

\bibitem{BrickellS08:DataMiningUtility}
Justin Brickell and Vitaly Shmatikov.
\newblock The cost of privacy: destruction of data-mining utility in anonymized
  data publishing.
\newblock In {\em KDD}, 2008.

\bibitem{ChakrabartiKSB08}
Soumen Chakrabarti, Rajiv Khanna, Uma Sawant, and Chiru Bhattacharyya.
\newblock Structured learning for non-smooth ranking losses.
\newblock In {\em KDD}, pages 88--96, 2008.

\bibitem{DworkKMMN06_OurDataOurselves}
Cynthia Dwork, Krishnaram Kenthapadi, Frank McSherry, Ilya Mironov, and Moni
  Naor.
\newblock Our data, ourselves: Privacy via distributed noise generation.
\newblock In {\em EUROCRYPT}, 2006.

\bibitem{dwork06Calibrating}
Cynthia Dwork, Frank McSherry, Kobbi Nissim, and Adam Smith.
\newblock Calibrating noise to sensitivity in private data analysis.
\newblock In {\em TCC}, 2006.

\bibitem{GoetzMWXG09:searchlogsV2}
Michaela G{\"o}tz, Ashwin Machanavajjhala, Guozhang Wang, Xiaokui Xiao, and
  Johannes Gehrke.
\newblock Privacy in search logs.
\newblock {\em CoRR}, abs/0904.0682v2, 2009.

\bibitem{cite:book/DataMining/2000}
Jiawei Han and Micheline Kamber.
\newblock {\em Data Mining: Concepts and Techniques}.
\newblock Morgan Kaufmann, 1st edition, September 2000.

\bibitem{HeN09:HierarchicalAnon}
Yeye He and Jeffrey~F. Naughton.
\newblock Anonymization of set-valued data via top-down, local generalization.
\newblock {\em PVLDB}, 2(1):934--945, 2009.

\bibitem{YuanAnonymization09}
Yuan Hong, Xiaoyun He, Jaideep Vaidya, Nabil Adam, and Vijayalakshmi Atluri.
\newblock Effective anonymization of query logs.
\newblock In {\em CIKM}, 2009.

\bibitem{JonesKPT07:PrivacyQueryLogs}
Rosie Jones, Ravi Kumar, Bo~Pang, and Andrew Tomkins.
\newblock "{I} know what you did last summer": query logs and user privacy.
\newblock In {\em CIKM}, 2007.

\bibitem{JonesRMG06:QuerySubstitution}
Rosie Jones, Benjamin Rey, Omid Madani, and Wiley Greiner.
\newblock Generating query substitutions.
\newblock In {\em WWW}, 2006.

\bibitem{KasiviswanathanLNRS08}
Shiva~Prasad Kasiviswanathan, Homin~K. Lee, Kobbi Nissim, Sofya Raskhodnikova,
  and Adam Smith.
\newblock What can we learn privately?
\newblock In {\em FOCS}, pages 531--540, 2008.

\bibitem{KorolovaKMN09:PrivateQueries}
Aleksandra Korolova, Krishnaram Kenthapadi, Nina Mishra, and Alexandros
  Ntoulas.
\newblock Releasing search queries and clicks privately.
\newblock In {\em WWW}, 2009.

\bibitem{KumarNPT07:AnonQueryLogs}
Ravi Kumar, Jasmine Novak, Bo~Pang, and Andrew Tomkins.
\newblock On anonymizing query logs via token-based hashing.
\newblock In {\em WWW}, 2007.

\bibitem{Luo:SurveyItemsets09}
Yongcheng Luo, Yan Zhao, and Jiajin Le.
\newblock A survey on the privacy preserving algorithm of association rule
  mining.
\newblock {\em Electronic Commerce and Security, International Symposium},
  1:241--245, 2009.

\bibitem{ashwin08:map}
Ashwin Machanavajjhala, Daniel Kifer, John~M. Abowd, Johannes Gehrke, and Lars
  Vilhuber.
\newblock Privacy: Theory meets practice on the map.
\newblock In {\em ICDE}, 2008.

\bibitem{MotwaniN:searchlogs}
Rajeev Motwani and Shubha Nabar.
\newblock Anonymizing unstructured data.
\newblock {\em arXiv}, 2008.

\bibitem{NissimRS07:Smooth}
Kobbi Nissim, Sofya Raskhodnikova, and Adam Smith.
\newblock Smooth sensitivity and sampling in private data analysis.
\newblock In {\em STOC}, 2007.

\bibitem{Samarati01}
Pierangela Samarati.
\newblock Protecting respondents' identities in microdata release.
\newblock {\em IEEE Trans. on Knowl. and Data Eng.}, 13(6):1010--1027, 2001.

\end{thebibliography}

\appendix
\section{Online Appendix}
This appendix is available online~\cite{GoetzMWXG09:searchlogsV2}. We provide it here for the convenience of the reviewers. It is not meant to be part of the final paper.

\subsection{Analysis of ZEALOUS: Proof of Theorem 10} \label{app:analysis}
 
	Let $H$ be the keyword histogram constructed by ZEALOUS in Step 2  when applied to $S$  and $K$ be the set of keywords in $H$ whose count equals $\tau$. Let $\Omega$ be the set of keyword histograms, that do not contain any keyword in $K$. For notational simplicity, let us denote ZEALOUS as a function $Z$. We will prove the theorem by showing that, given Equations (\ref{eqn:zealous-lambda}) and (\ref{eqn:zealous-tau}),
\begin{equation} \label{eqn:zealous-condition1}
Pr[Z(S) \notin \Omega] \le \delta,
\end{equation}
and for any keyword histogram $\omega \in \Omega$ and for any neighboring search log $S'$ of $S$,
\begin{equation} \label{eqn:zealous-condition2}
e^{-\epsilon} \cdot Pr[Z(S') \!=\! \omega] \le Pr[Z(S)\! = \!\omega] \le e^{\epsilon} \cdot Pr[Z(S')\! =\! \omega].
\end{equation}

We will first prove that Equation (\ref{eqn:zealous-condition1}) holds. Assume that the $i$-th keyword in $K$ has a count $\tilde{c}_i$ in $Z(S)$ for $i \in [1, |K|]$. Then,
\begin{eqnarray}
\lefteqn{Pr[Z(S) \notin \Omega]} \nonumber \\ 
&=& Pr\Big[\exists i \in [1, |K|], \tilde{c}_i > \tau'\Big] \nonumber \\
&=& 1 - Pr\Big[\forall i \in [1, |K|], \tilde{c}_i \le \tau' \Big] \nonumber \\
&=& 1 - \prod_{i \in [1, |K|]} \left(\int_{-\infty}^{\tau'-\tau} \frac{1}{2 \lambda} e^{-\frac{|x|}{\lambda}}dx \right) \nonumber \\
& & \textrm{(the noise added to $\tilde{c}_i$ has to be $\geq \tau'-\tau$)} \nonumber \\
&=& 1 - \left(1 - \frac{1}{2} \cdot e^{-\frac{\tau' - \tau}{\lambda}}\right)^{|K|} \nonumber \\
&\le& \frac{|K|}{2} \cdot e^{-\frac{\tau' - \tau}{\lambda}} \nonumber \\
&\le& \frac{U \cdot m}{2\tau} \cdot e^{-\frac{\tau' - \tau}{\lambda}} \qquad \textrm{(because $|K| \le U\cdot m / \tau$)} \nonumber \\
& &  \nonumber \\
&\le& \frac{U \cdot m}{2\tau} \cdot e^{-\frac{-\lambda \ln\left(\frac{2\delta}{U \cdot m / \tau}\right)}{\lambda}} \qquad \textrm{(by Equation~\ref{eqn:zealous-tau})} \nonumber \\
&=& \delta.
\end{eqnarray}

Next, we will show that Equation (\ref{eqn:zealous-condition2}) also holds. Let $S'$ be any neighboring search log of $S$. Let $\omega$ be any possible output of ZEALOUS given $S$, such that $\omega \in \Omega$. To establish Equation (\ref{eqn:zealous-condition2}), it suffices to prove that
\begin{align}
&\frac{Pr[Z(S)=\omega]}{Pr[Z(S')=\omega]} \le e^{\epsilon}, \text{ and} \label{eqn:zealous-ratio-1} \\
&\frac{Pr[Z(S')=\omega]}{Pr[Z(S)=\omega]} \le e^{\epsilon}\label{eqn:zealous-ratio-2}.
\end{align}

We will derive Equation (\ref{eqn:zealous-ratio-1}). The proof of
(\ref{eqn:zealous-ratio-2}) is analogous.

Let $H'$ be the keyword histogram constructed by ZEALOUS in Step 2  when applied to $S'$. Let $\Delta$ be the set of keywords that have different counts in $H$ and $H'$. Since $S$ and $S'$ differ in the search history of a single user, and each user contributes at most $m$ keywords, we have $|\Delta| \le 2m$. Let $k_i$ $(i \in [1, |\Delta|])$ be the $i$-th keyword in $\Delta$, and $d_i$, $d_i'$, and $d_i^*$ be the counts of $k_i$ in $H$, $H'$, and $\omega$, respectively. Since a user adds at most one to the count of a keyword (see Step 2.), we have $d_i - d_i' = 1$ for any $i \in [1, |\Delta|]$. To simplify notation, let $E_i$, $E_i'$, and ${E_i}^*,{E_i'}^* $ denote the event that $k_i$ has counts $d_i$, $d_i'$, $d_i^*$ in $H$, $H'$, and $Z(S), Z(S')$, respectively. Therefore,
\begin{equation} \nonumber
\frac{Pr[Z(S)=\omega]}{Pr[Z(S')=\omega]} = \prod_{i \in [1, |\Delta|]} \frac{Pr[{E_i}^* \mid E_i]}{Pr[{E_i'}^* \mid E_i']}.
\end{equation}

In what follows, we will show that $\frac{Pr[{E_i}^* \mid E_i]}{Pr[{E_i'}^* \mid E_i']} \le e^{1/\lambda}$ for any $i \in [1, |\Delta|]$. We differentiate three cases: (i) $d_i \geq \tau$, $d_i^* \geq \tau$, (ii)  $d_i < \tau$ and (iii) $d_i = \tau$ and $d_i^* =\tau-1$.

Consider  case (i) when $d_i$ and $ d_i^*$ are at least $\tau$. 
Then, if $d_i^* > 0$, we have
\begin{eqnarray}
\lefteqn{\frac{Pr[{E_i}^* \mid E_i]}{Pr[{E_i'}^* \mid E_i']} } \nonumber \\
&=& \frac{\frac{1}{2\lambda} e^{-|d_i^* - d_i|/\lambda}}{\frac{1}{2\lambda} e^{-|d_i^* - d_i'|/\lambda}} \nonumber \\
&=& e^{(|d_i^* - d_i'| - |d_i^* - d_i|)/\lambda} \nonumber \\
&\le& e^{|d_i - d_i'|/\lambda} \nonumber \\
&=& e^\frac{1}{\lambda}. \qquad \textrm{(because $|d_i - d_i'| = 1$ for any $i$)}\nonumber
\end{eqnarray}
On the other hand, if $d_i^* =0$,
\begin{eqnarray}
\frac{Pr[{E_i}^* \mid E_i]}{Pr[{E_i'}^* \mid E_i']} 
= \frac{\int_{-\infty}^{\tau' - d_i}\frac{1}{2\lambda} e^{-|x|/\lambda}dx}{\int_{-\infty}^{\tau' - d_i'}\frac{1}{2\lambda} e^{-|x|/\lambda}dx} 
\le e^\frac{1}{\lambda}. \nonumber
\end{eqnarray}

Now consider case  (ii) when $d_i $ is less than $ \tau$. Since $\omega \in \Omega$, and ZEALOUS eliminates all counts in $H$ that are smaller than $\tau$, we have $d_i^* = 0$, and $Pr[E_i^* \mid E_i] = 1$. On the other hand,
\begin{equation} \nonumber \\
Pr[{E_i'}^* \mid E_i'] =
\left\{
\begin{array}{ll}
1, & \textrm{if $d_i' \le \tau$} \\
1 - \frac{1}{2} e^{-|\tau' - d_i'|/\lambda}, & \textrm{otherwise}
\end{array}
\right.
\end{equation}
Therefore,
\begin{eqnarray}
\lefteqn{\frac{Pr[{E_i}^* \mid E_i]}{Pr[{E_i'}^* \mid E_i']} } \nonumber \\
&\le& \frac{1}{1 - \frac{1}{2} e^{-|\tau' - d_i'|/\lambda}} \nonumber \\
&\le& \frac{1}{1 - \frac{1}{2} e^{-(\tau' - \tau)/\lambda}} \nonumber \\
&\le& \frac{1}{1 - \frac{1}{2} e^{\ln\left(2-2e^{-\frac{1}{\lambda}}\right)}} \qquad \textrm{(by Equation~\ref{eqn:zealous-condition2})}\nonumber \\
&=& e^\frac{1}{\lambda}. \nonumber
\end{eqnarray}

Consider now case (iii) when $d_i = \tau$ and $d_i^* = \tau-1$. Since $\omega \in \Omega$  we have $d_i^* = 0$. Moreover, since ZEALOUS eliminates all counts in $H$ that are smaller than $\tau$, it follows that $Pr[E_i^* \mid E_i'] = 1$. Therefore,
\begin{eqnarray}
\frac{Pr[{E_i}^* \mid E_i]}{Pr[{E_i'}^* \mid E_i']} 
= Pr[{E_i}^* \mid E_i]
\le e^\frac{1}{\lambda}. \nonumber
\end{eqnarray}

In summary, $\frac{Pr[{E_i}^* \mid E_i]}{Pr[{E_i'}^* \mid E_i']}  \le e^{1/\lambda}$. Since $|\Delta| \le 2m$, we have
\begin{eqnarray} \nonumber
\lefteqn{\frac{Pr[Z(S)=\omega]}{Pr[Z(S')=\omega]}} \nonumber \\
&=& \prod_{i \in [1, |\Delta|]} \frac{Pr[{E_i}^* \mid E_i]}{Pr[{E_i'}^* \mid E_i']}  \nonumber \\
&\le& \prod_{i \in [1, |\Delta|]} e^{1/\lambda} \nonumber \\
&=& e^{|\Delta|/\lambda} \nonumber \\
&\le& e^{\epsilon} \qquad \textrm{(by Equation~\ref{eqn:zealous-condition1} and $|\Delta| \le 2m$)}. \nonumber
\end{eqnarray}
This concludes the proof of the theorem.

\subsection{Proof of Proposition 1}\label{app:comp_impl}
Assume that, for all search logs $S$, we can divide the output space $\Omega$ into to two sets $\Omega_1, \Omega_2 $, such that
	\[(1) \Pr[\mathcal{A}(S) \in \Omega_2] \leq  \delta, \text{ and }\]
for all search logs $S'$ differing from $S$ only in the search history of a single user and for all $O \in \Omega_1$:
	\begin{align*}
	(2) & \Pr[\San(S)=O] \leq  e^{\epsilon} \Pr[\San(S')=O] \text{ and }  \\
	& \Pr[\San(S')=O] \leq e^{\epsilon} \Pr[\San(S)=O].
	\end{align*}
Consider any subset $\mathcal{O}$ of the output space $\Omega$ of $\San$. Let $\mathcal{O}_1 = \mathcal{O} \cap \Omega_1$ and $\mathcal{O}_2 = \mathcal{O} \cap \Omega_2$. We have
\begin{eqnarray}
\lefteqn{\Pr[\San(S) \in \mathcal{O}]} \nonumber \\
 &=& \hspace{-3mm}\int_{O\in \mathcal{O}_2} \hspace{-3mm}\Pr[\San(S) = O] dO + \int_{O \in \mathcal{O}_1} \hspace{-3mm}\Pr[\San(S) = O] dO \nonumber \\
&\le&  \hspace{-3mm}\int_{O \in \Omega_2} \hspace{-3mm}\Pr[\San(S) = O] dO + e^{\epsilon}\hspace{-2mm} \int_{O\in \Omega_1}\hspace{-3mm} \Pr[\San(S') = O] dO  \nonumber \\
&\le& \delta + e^{\epsilon} \int_{O \in \Omega_1} \Pr[\San(S') = O] dO \nonumber \\
&\le& \delta + e^{\epsilon} \cdot \Pr[\San(S') \in \Omega_1]. \nonumber
\end{eqnarray}

\subsection{Proof of Proposition 2}\label{app:comp_alg}

	We have to show that for all search logs $S, S'$ differing in one user history  and for all sets $\mathcal{O}:$
	\[ \Pr[\hat{\mathcal{A}}(S) \in \mathcal{O}] \leq \Pr[\hat{\mathcal{A}}(S') \in \mathcal{O}] + 1/(|\mathcal{D}|-1). \]
	Since Algorithm $\hat{\mathcal{A}}$ neglects all but the first input this is true for neighboring search logs not differing in the first user's input. We are left with the case of two neighboring search logs $S, S'$  differing in the search history of the first user. Let us analyze the output distributions of Algorithm 1 under these two inputs $S$ and $ S'$. For all search histories except the search histories of the first user in $S, S'$ the output probability is  $1/(|\mathcal{D}|-1)$ for either input.  Only for the two  search histories of the first user $S_1, S'_1$ the output probabilities differ: Algorithm 1 never outputs $S_1$ given $S$, but it outputs this search history with probability $1/(|\mathcal{D}|-1)$  given $S'$. Symmetrically, Algorithm $\hat{\mathcal{A}}$ never outputs $S'_1$ given $S'$, but it outputs this search history with probability $1/(|\mathcal{D}|-1)$  given $S$.
	Thus, we have for all sets $\mathcal{O}$
	\begin{align}
		 \Pr[\hat{\mathcal{A}}(S) \in \mathcal{O}] 
		&= \sum_{d \in \mathcal{O} \cap (\mathcal{D}-S_1)} 1/(|\mathcal{D}|-1)\\
		&\leq 1/(|\mathcal{D}|-1) + \sum_{d \in \mathcal{O} \cap (\mathcal{D}-S_2)} 1/(|\mathcal{D}|-1)\\
		&=  \Pr[\hat{\mathcal{A}}(S) \in \mathcal{O}]  + 1/(|\mathcal{D}|-1)
	\end{align}

\end{sloppy}
\end{document}